\documentclass{amsart}
\pdfoutput=1
\usepackage[foot]{amsaddr}
\usepackage{amsmath}
\usepackage{amsthm}
\usepackage{amssymb}
\usepackage{graphicx}
\usepackage{url}
\usepackage[hidelinks]{hyperref}
\usepackage{hhline}

\newcommand{\Sym}{\ensuremath{\mathrm{Sym}}}
\newcommand{\Aut}{\ensuremath{\mathrm{Aut}}}
\newcommand{\Iso}{\ensuremath{\mathrm{Iso}}}
\newcommand{\erel}[2]{\ensuremath{{#1}\kern -0.5pt e\kern 0.5pt {#2}}}
\newcommand{\False}{\mathbf{f}}
\newcommand{\True}{\mathbf{t}}

\newcommand{\NamedIndent}[2]{
\begin{center}
\begin{tabular}{@{}l@{\hspace*{3mm}}p{10.5cm}}
{#1} & {#2}
\end{tabular}
\end{center}
}

\newtheorem{lemma}{Lemma}
\theoremstyle{remark}
\newtheorem{example}{Example}
\theoremstyle{definition}

\begin{document}

\title[Adaptive prefix-assignment for symmetry reduction]{An adaptive prefix-assignment technique\\{} for symmetry reduction}

\author{Tommi Junttila$^\dag$}
\author{Matti Karppa$^\dag$}
\author{Petteri Kaski$^\dag$}
\author{Jukka Kohonen$^{\dag \ddag}$}
\address{$^\dag$ Aalto University, Department of Computer Science}
\email{Tommi.Junttila@aalto.fi}
\email{Matti.Karppa@aalto.fi}
\email{Petteri.Kaski@aalto.fi}
\address{$^\ddag$ \textit{Current address:} University of Helsinki, Department of Computer Science}
\email{Jukka.Kohonen@cs.helsinki.fi}

\begin{abstract}
This paper presents a technique for symmetry reduction that adaptively
assigns a prefix of variables in a system of constraints so that the
generated prefix-assignments are pairwise nonisomorphic under the
action of the symmetry group of the system. 
The technique is based on McKay's canonical extension framework 
[J.~Algorithms 26 (1998), no.~2, 306--324].
Among key features of the technique are 
(i) adaptability---the prefix sequence can be user-prescribed and truncated for compatibility with the group of symmetries; 
(ii) parallelizability---prefix-assignments can be processed in parallel independently of each other; 
(iii) versatility---the method is applicable whenever the group of symmetries can be concisely represented as the automorphism group of a vertex-colored graph; and
(iv) implementability---the method can be implemented relying on a canonical labeling map for vertex-colored graphs as the only nontrivial subroutine.
To demonstrate the practical applicability of our technique,
we have prepared an experimental open-source implementation of the technique 
and carry out a set of experiments that demonstrate ability to reduce symmetry 
on hard instances. Furthermore, we demonstrate that the implementation 
effectively parallelizes to compute clusters with multiple nodes via 
a message-passing interface.
\end{abstract}

\maketitle


\section{Introduction}

\subsection{Symmetry reduction}
Systems of constraints often have substantial symmetry. For example, 
consider the following system of Boolean clauses:
\begin{equation}
\label{eq:example}
(x_1\vee x_2)
\wedge
(x_1\vee \bar x_3\vee \bar x_5)
\wedge
(x_2\vee \bar x_4\vee \bar x_6)\,.
\end{equation}
The associative and commutative symmetries of disjunction and conjunction
induce symmetries between the variables of \eqref{eq:example}, a fact that 
can be captured by stating that the group $\Gamma$ generated by
the two permutations
$(x_1\ x_2)(x_3\ x_4)(x_5\ x_6)$ and
$(x_4\ x_6)$ consists of all permutations of the variables 
that map \eqref{eq:example} to itself. That is, $\Gamma$ is the 
\emph{automorphism group} of the system \eqref{eq:example},
cf.~Section~\ref{sect:symmetry}.

Known symmetry in a constraint system is a great asset from the 
perspective of solving the system, in particular since symmetry 
enables one to disregard partial solutions that are \emph{isomorphic} to 
each other under the action of $\Gamma$ on the space of partial solutions. 
Techniques for such \emph{isomorph rejection}%
\footnote{%
A term introduced by J.~D.~Swift~\cite{Swift:1960}; 
cf.~Hall and Knuth~\cite{Hall:1965} for 
a survey on early work on exhaustive computer search and combinatorial
analysis.}{}~\cite{Swift:1960}
(alternatively, \emph{symmetry reduction} or \emph{symmetry breaking}) 
are essentially mandatory if one desires an exhaustive traversal of 
the (pairwise nonisomorphic) solutions of a highly symmetric system of 
constraints, or if the system is otherwise difficult to solve, for example, 
with many ``dead-end'' partial solutions compared with the actual
number of solutions.

A prerequisite to symmetry reduction is that the symmetries are known. 
In many cases it is possible to automatically discover and compute these 
symmetries to enable practical and automatic symmetry reduction.
In this context the dominant computational approach for combinatorial systems
of constraints is to represent $\Gamma$ via the automorphism group of 
a vertex-colored graph that captures the symmetries in the system. 
Carefully engineered tools for working with symmetries of vertex-colored 
graphs~\cite{DargaLiffitonSakallahMarkov:2004,JunttilaKaski:2007,McKay:1981,McKayPiperno:2014} and 
permutation group algorithms~\cite{Butler:1991,Seress:2003} 
then enable one to perform symmetry reduction. For example, for purposes 
of symmetry computations we may represent~\eqref{eq:example} as the 
following vertex-colored graph:

\begin{equation}
\label{eq:example-graph}
\begin{array}{c}
\includegraphics[scale=0.8]{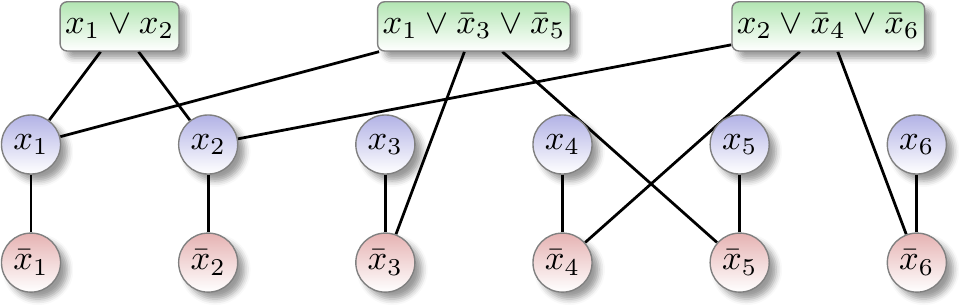}
\end{array}
\end{equation}

\noindent
In particular, the graph representation~\eqref{eq:example-graph} enables 
us to discover and reduce symmetry to
avoid redundant work when solving the underlying system~\eqref{eq:example}.

\subsection{Our contribution}
The objective of this paper is to present a novel technique for symmetry 
reduction on systems of constraints. The technique is based 
on adaptively assigning values to a prefix of the variables so that the 
obtained prefix-assignments are pairwise nonisomorphic under the action 
of~$\Gamma$. The technique can be seen as an instantiation of 
McKay's~\cite{McKay:1998} influential canonical extension framework for 
isomorph-free exhaustive generation.

To give a brief outline of the technique, 
suppose we are working with a system of constraints over a finite set $U$ 
of variables that take values in a finite set $R$. Suppose furthermore that 
$\Gamma\leq\Sym(U)$ is the automorphism group of the system. 
Select $k$ distinct variables $u_1,u_2,\ldots,u_k$ in $U$. 
These $k$ variables form the \emph{prefix sequence} considered by 
the method. The technique works by assigning values in $R$ to the variables of 
the prefix, in prefix-sequence order, with $u_1$ assigned first, then $u_2$, 
then $u_3$, and so forth, so that at each step the partial assignments 
so obtained are pairwise nonisomorphic under the action of $\Gamma$.
For example, in \eqref{eq:example} the partial assignments
$x_1\mapsto 0,\,x_2\mapsto 1$ and $x_1\mapsto 1,\,x_2\mapsto 0$ are
isomorphic since $(x_1\ x_2)(x_3\ x_4)(x_5\ x_6)\in\Gamma$ 
maps one assignment onto the other; in total there are three nonisomorphic
assignments to the prefix $x_1,x_2$ in \eqref{eq:example}, namely 
(i) $x_1\mapsto 0,\,x_2\mapsto 0$, 
(ii) $x_1\mapsto 0,\,x_2\mapsto 1$, 
and (iii) $x_1\mapsto 1,\,x_2\mapsto 1$.
Each partial assignment that represents
an isomorphism class can then be used to reduce redundant work when 
solving the underlying system by standard techniques---%
in the nonincremental case, the system is augmented with
a symmetry-breaking predicate requiring that one of the nonisomorphic
partial assignments holds, while in the incremental setting \cite{HeuleEtAl:HVC2011,Wieringa:inccnf}
the partial assignments can be solved independently or even in parallel.

Our contribution in this paper lies in how the isomorph rejection is 
implemented at the level of isomorphism classes of partial assignments
by careful reduction to McKay's~\cite{McKay:1998} isomorph-free exhaustive
generation framework. The key 
technical contribution is that we observe how to generate the partial 
assignments in a normalized form that enables both \emph{adaptability}
(that is, the prefix $u_1,u_2,\ldots,u_k$ can be arbitrarily selected 
to match the structure of $\Gamma$) and precomputation of 
the extending variable-value orbits along a prefix.

Among further key features of the technique are:
\begin{enumerate}
\item
\emph{Implementability.}
The technique can be implemented by relying on a canonical labeling map 
for vertex-colored graphs 
(cf.~\cite{JunttilaKaski:2007,McKay:1981,McKayPiperno:2014})
as the only nontrivial subroutine that is invoked once
for each partial assignment considered.%
\item
\emph{Versatility.}
The method is applicable whenever the group of symmetries can be 
concisely represented as a vertex-colored graph; 
cf.~\eqref{eq:example} and \eqref{eq:example-graph}. This is useful in 
particular when the underlying system has symmetries that are not easily 
discoverable from the final constraint encoding, for example, due to the fact 
that the constraints have been compiled or optimized%
\footnote{For a beautiful illustration, we refer to Knuth's \cite[\S7.1.2, Fig.~10]{Knuth:2011} example of optimum Boolean chains for 5-variable symmetric Boolean functions---from each optimum chain it is far from obvious that the chain represents a symmetric Boolean function. (See also Example~\ref{exa:brent}.)}{} from a higher-level representation in a symmetry-obfuscating manner.
A graphical representation can represent such symmetry directly and 
independently of the compiled/optimized form of the system.
\item
\emph{Parallelizability.}
As a corollary of implementing McKay's~\cite{McKay:1998} framework,
the technique does not need to store representatives of isomorphism classes 
in memory to perform isomorph rejection, which enables easy parallelization 
since the partial assignments can be processed independently of each other. 
\end{enumerate}
The required mathematical preliminaries
on symmetry and McKay's framework are reviewed in
Sections~\ref{sect:symmetry}~and~\ref{sect:mckay}, respectively. 
The main technical contribution of this paper is developed in
Section~\ref{sect:prefix-method} where we present the
prefix-assignment technique that we will subsequently extend to account for value symmetry in
Section~\ref{sect:value-symmetry}.
Our development in Sections~\ref{sect:prefix-method} and~\ref{sect:value-symmetry} relies on an abstract
group $\Gamma$, with the understanding that a concrete implementation
can be designed e.g.~in terms of a vertex-colored graph
representation, as will be explored in
Section~\ref{sect:colored-graphs}.  

To demonstrate the practical applicability of our technique, we have
prepared an open-source parallel implementation~\cite{ReduceGithub:2018}.
The implementation is structured as a preprocessor
that works with an explicitly given graph representation and utilizes
the \emph{nauty}~\cite{McKay:1981,McKayPiperno:2014}{} canonical
labeling software for vertex-colored graphs as a subroutine to prepare
an exhaustive collection of nonisomorphic prefix assignments relative
to a user-supplied prefix of variables, and the Message Passing 
Interface (MPI) for parallelization. Further details of the
implementation are presented in Section~\ref{sect:implementation}.  In
Section~\ref{sect:experiments}, we report on a set of experiments that 
(i) demonstrate the ability to reduce symmetry on hard instances, 
(ii) study the serendipity of an auxiliary graph for encoding the symmetries
in an instance, 
(iii) give examples of instances with hard combinatorial symmetry where 
our technique performs favorably in comparison with earlier techniques, and 
(iv) study the parallel speedup obtainable when we distribute 
the symmetry reduction task to a compute cluster with multiple compute nodes. 

\subsection{Earlier work}
A classical way to exploit symmetry in a system of constraints is 
to augment the system with so-called symmetry-breaking predicates (SBP)
that eliminate either some or all symmetric solutions \cite{AloulEtAl:IJCAI2003,CrawfordEtAl:KR1996,GentEtAl:Handbook2006,Sakallah:Handbook2009}.
Such constraints are typically lexicographic leader (lex-leader) constraints 
that are derived from a generating set for the group of symmetries $\Gamma$. 
Among recent work in this area, 
Devriendt \emph{et al.}~\cite{DevriendtEtAl:SAT2016} extend the approach 
by presenting a more compact way for expressing SBPs
and a method for detecting ``row interchangeabilities''.
Itzhakov and Codish~\cite{ItzhakovCodish:2016} present a method for 
finding a set of symmetries whose corresponding lex-leader constraints 
are enough to completely break symmetries in search problems on small
(10-vertex) graphs; this approach is extended by Codish \emph{et al.}~\cite{CodishEtAl:2016} by adding pruning predicates that simulate the 
first iterations of the equitable partition refinement algorithm 
of \emph{nauty}~\cite{McKay:1981,McKayPiperno:2014}.
Heule~\cite{Heule:SYNASC2016} shows that small complete 
symmetry-breaking predicates can be computed by considering arbitrary 
Boolean formulas instead of lex-leader formulas.

Our present technique can be seen as a method for producing 
symmetry-breaking predicates by augmenting the system of constraints with 
the disjunction of the nonisomorphic partial assignments.
The main difference to the related work above is that our technique does not
produce the symmetry-breaking predicate from a set of generators for 
$\Gamma$ but rather the predicate is produced recursively, and with the
possibility for parallelization, by classifying orbit representatives up 
to isomorphism using McKay's~\cite{McKay:1998} framework. As such our 
technique breaks all symmetry with respect to the prescribed prefix,
but comes at the cost of additional invocations of graph-automorphism and 
canonical-labeling tools. This overhead and increased symmetry reduction
in particular means that our technique is best suited for constraint 
systems with hard combinatorial symmetry that is not easily capturable
from a set of generators, such as symmetry in combinatorial classification
problems~\cite{KaskiOstergard:2006}. In addition to McKay's \cite{McKay:1998} 
canonical extension framework, other standard frameworks for isomorph-free 
exhaustive generation in this context include \emph{orderly algorithms} due to
Farad{\v z}ev~\cite{Faradzev:1978} and Read \cite{Read:1978},
as well as
the homomorphism principle for group actions due to
Kerber and  Laue~\cite{Kerber:1998}.

It is also possible to break symmetry within a constraint solver during 
the search by dynamically adding constraints that rule out symmetric parts of 
the search space (cf.~\cite{ChuEtAl:Constraints2014,GentEtAl:Handbook2006}).
If we use the nonisomorphic partial assignments produced by our technique 
as assumption sequences (cubes) in the incremental cube-and-conquer approach
\cite{HeuleEtAl:HVC2011,Wieringa:inccnf},
our technique can be seen as a restricted way of breaking the symmetries
in the beginning of the search, with the 
benefit---as with cube-and-conquer---that 
the portions of the search space induced by the partial 
assignments can be solved in parallel, either with complete independence 
or with appropriate sharing of information (such as conflict clauses) 
between the parallel nodes executing the search. For further work in dynamic 
symmetry breaking, cf.~\cite{Devriendt:2017,Benhamou:1994,Benhamou:2010,Sabharwal:2009,Schaafsma:2009,Benhamou:2010b,Devriendt:2012}.

For work on isomorphism and canonical labeling techniques, 
cf.~\cite{Babai2016,GroheNSW2018,LokshtanovPPS17,McKayPiperno:2014}.

\section{Preliminaries on group actions and symmetry}
\label{sect:symmetry}

This section reviews relevant mathematical preliminaries and notational
conventions for groups, group actions, symmetry, and isomorphism for our 
subsequent development. 
(Cf.~\cite{Butler:1991,Dixon:1996,Humphreys:1996,KaskiOstergard:2006,Kerber:1999,Seress:2003} for further reference.)

\subsection{Groups and group actions}
Let $\Gamma$ be a finite group and let $\Omega$ be a finite set 
(the domain) on which $\Gamma$ acts. For two groups $\Lambda$ and $\Gamma$, 
let us write $\Lambda\leq\Gamma$ to indicate that $\Lambda$ is a subgroup 
of $\Gamma$. We use exponential notation for group actions, and accordingly
our groups act from the right. That is, for an object $X\in\Omega$ 
and $\gamma\in \Gamma$, let us write $X^\gamma$ for the object in $\Omega$ 
obtained by acting on $X$ with $\gamma$. Accordingly, we have $X^{(\beta\gamma)}=(X^\beta)^\gamma$ for 
all $\beta,\gamma\in\Gamma$ and $X\in\Omega$. For a finite set $V$, let us
write $\Sym(V)$ for the group of all permutations of $V$ with composition
of mappings as the group operation.

\subsection{Orbit and stabilizer, the automorphism group}
For an object $X\in\Omega$ let us write 
$X^\Gamma=\{X^\gamma:\gamma\in\Gamma\}$ 
for the \emph{orbit} of $X$ under the action of $\Gamma$ and 
$\Gamma_X=\{\gamma\in \Gamma:X^\gamma=X\}\leq\Gamma$ 
for the \emph{stabilizer} subgroup of $X$ in $\Gamma$. 
Equivalently we say that $\Gamma_X$ is the \emph{automorphism group} of $X$ 
and write $\Aut(X)=\Gamma_X$ whenever $\Gamma$ is clear from the context;
if we want to stress the acting group we write $\Aut_\Gamma(X)$.
 
We write $\Omega/\Gamma=\{X^\Gamma:X\in\Omega\}$ for the set of all
orbits of $\Gamma$ on $\Omega$.  For $\Lambda\leq\Gamma$ and
$\gamma\in\Gamma$, let us write
$\Lambda^{\gamma}=\gamma^{-1}\Lambda\gamma
=\{\gamma^{-1}\lambda\gamma:\lambda\in\Lambda\} \leq\Gamma$ for the
$\gamma$-\emph{conjugate} of $\Lambda$.  For all $X\in\Omega$ and
$\gamma\in\Gamma$ we have $\Aut(X^\gamma)=\Aut(X)^\gamma$.  That is,
the automorphism groups of objects in an orbit are conjugates of each
other.

\subsection{Isomorphism}
We say that two objects are \emph{isomorphic} if they are on 
the same orbit of $\Gamma$ in $\Omega$. 
In particular, $X,Y\in\Omega$ are isomorphic 
if and only if there exists an \emph{isomorphism} $\gamma\in\Gamma$ 
from $X$ to $Y$ that satisfies $Y=X^\gamma$. An isomorphism from an
object to itself is an \emph{automorphism}. 
Let us write $\Iso(X,Y)$ for the set of all isomorphisms from $X$ to $Y$. 
We have that $\Iso(X,Y)=\Aut(X)\gamma=\gamma\Aut(Y)$
where $\gamma\in\Iso(X,Y)$ is arbitrary. 
Let us write $X\cong Y$ to indicate that $X$ and $Y$ are isomorphic. 
If we want to stress the group $\Gamma$ under whose action isomorphism 
holds, we write $X\cong_\Gamma Y$.

\subsection{Elementwise action on tuples and sets} 
Suppose that $\Gamma$ acts on two sets, $\Omega$ and $\Sigma$.
We extend the action to the Cartesian product $\Omega\times\Sigma$ 
elementwise by defining $(X,S)^\gamma=(X^\gamma,S^\gamma)$ 
for all $(X,S)\in\Omega\times\Sigma$ and $\gamma\in \Gamma$. 
Isomorphism extends accordingly; for example, we say that $(X,S)$ and $(Y,T)$ 
are isomorphic and write $(X,S)\cong (Y,T)$ if there exists a 
$\gamma\in \Gamma$ with $Y=X^\gamma$ and $T=S^\gamma$. 
Suppose that $\Gamma$ acts on a set $U$. 
We extend the action of $\Gamma$ on $U$ to an elementwise action of $\Gamma$ 
on subsets $W\subseteq U$ by setting $W^\gamma=\{w^\gamma:w\in W\}$ for all 
$\gamma\in \Gamma$ and $W\subseteq U$. In what follows we will tacitly work
with these elementwise actions on tuples and sets unless explicitly otherwise 
indicated.

\subsection{Canonical labeling and canonical form}
A function $\kappa:\Omega\rightarrow \Gamma$ is a 
\emph{canonical labeling map} for the action of $\Gamma$ on $\Omega$ if 
\begin{center}
\begin{tabular}{@{}l@{\hspace*{3mm}}p{10cm}}
(K) & 
for all $X,Y\in\Omega$ it holds that $X\cong Y$ implies 
$X^{\kappa(X)}=Y^{\kappa(Y)}$ (canonical labeling).
\end{tabular}
\end{center}
For $X\in\Omega$ we say that $X^{\kappa(X)}$ is 
the \emph{canonical form} of $X$ in $\Omega$. 
From (K) it follows that isomorphic objects have identical canonical forms, 
and the canonical labeling map gives an isomorphism that takes an object to 
its canonical form. 

We assume that the act of computing $\kappa(X)$ for a 
given $X$ produces as a side-effect a set of generators for the
automorphism group $\Aut(X)$.

\section{McKay's canonical extension method}

\label{sect:mckay}

This section reviews McKay's~\cite{McKay:1998} canonical extension method 
for isomorph-free exhaustive generation. Mathematically it will be convenient 
to present the method so that the isomorphism classes are captured as orbits 
of a group action of a group $\Gamma$, and extension takes place in one step 
from ``seeds'' to ``objects'' being generated, with the understanding that
the method can be applied inductively in multiple steps so that
the ``objects'' of the current step become the ``seeds'' for the next step. 
For completeness and ease of exposition, we also give a correctness proof for 
the method. We stress that all material in this section is
well known. (Cf.~\cite{KaskiOstergard:2006}.)

\subsection{Objects and seeds}
Let $\Omega$ be a finite set of \emph{objects} and let $\Sigma$ be a finite
set of \emph{seeds}. Let $\Gamma$ be a finite group that acts on $\Omega$ 
and $\Sigma$. Let $\kappa$ be a canonical labeling map for the action
of $\Gamma$ on $\Omega$. 

\subsection{Extending seeds to objects.}
Let us connect the objects and the seeds by means of a relation
$e\subseteq \Omega\times\Sigma$ that indicates which objects can be
built from which seeds by extension. For $X\in\Omega$ and $S\in\Sigma$
we say that $X$ \emph{extends} $S$ and write $\erel{X}{S}$ if
$(X,S)\in e$. We assume the relation $e$ satisfies
\begin{center}
\begin{tabular}{@{}l@{\hspace*{3mm}}p{10cm}}
(E1) & 
$e$ is a union of orbits of $\Gamma$, that is, $e^\Gamma=e$ (invariance), and\\[1mm]
(E2) &
for every object $X\in\Omega$ there exists a seed $S\in\Sigma$ such 
that $\erel{X}{S}$ (completeness).
\end{tabular}
\end{center}
For a seed $S\in\Sigma$, let us write
$e(S)=\{X\in\Omega:\erel{X}{S}\}$
for the set of all objects that extend $S$.

\subsection{Canonical extension}
We associate with each object a particular isomorphism-invariant
extension by which we want to extend the object from a seed.
A function $M:\Omega\rightarrow\Sigma$ is a \emph{canonical extension map} if 
\begin{center}
\begin{tabular}{@{}l@{\hspace*{3mm}}p{10cm}}
(M1) & 
for all $X\in\Omega$ it holds that $(X,M(X))\in e$
(extension), and\\[1mm]
(M2) &
for all $X,Y\in\Omega$ we have that $X\cong Y$ implies 
$(X,M(X))\cong (Y,M(Y))$
(canonicity).
\end{tabular}
\end{center}
That is, (M1) requires that $X$ is in fact an
extension of $M(X)$ and (M2) requires that isomorphic objects have isomorphic
canonical extensions. In particular, $X\mapsto (X,M(X))$ is 
a well-defined map from $\Omega/\Gamma$ to $e/\Gamma$.

\subsection{Generating objects from seeds}
Let us study the following procedure,
which is invoked
for exactly one representative $S\in\Sigma$ from each orbit in 
$\Sigma/\Gamma$:
\begin{center}
\begin{tabular}{@{}l@{\hspace*{3mm}}p{10cm}}
(P) & 
Let $S\in\Sigma$ be given as input. Iterate over all $X\in e(S)$. 
Perform zero or more isomorph rejection tests on $X$ and $S$. 
If the tests indicate we should accept $X$, visit $X$.
\end{tabular}
\end{center}

Let us first consider the case when there are no isomorph rejection tests. 

\begin{lemma}
\label{lem:mckay-completeness}
The procedure (P) visits every isomorphism 
class of objects in $\Omega$ at least once. 
\end{lemma}

\begin{proof}
To see that every isomorphism class is visited, let $Y\in\Omega$ be arbitrary. 
By (E2), there exists a $T\in\Sigma$ with $\erel{Y}{T}$. By our assumption on how
procedure (P) is invoked, $T$ is isomorphic to a unique $S$ such that 
procedure (P) is invoked with input $S$. Let $\gamma\in \Gamma$ be an 
associated isomorphism with $S^\gamma=T$. By (E1) and $\erel{Y}{T}$, 
we have $\erel{X}{S}$ for $X=Y^{\gamma^{-1}}$. By the structure of procedure (P) 
we observe that $X$ is visited and $X\cong Y$. Since $Y$ was arbitrary, 
all isomorphism classes are visited at least once.
\end{proof}

Let us next modify procedure (P) so that any two visits to the same
isomorphism class of objects originate from the same procedure invocation. 
Let $M:\Omega\rightarrow\Sigma$ be a canonical extension map.
Whenever we construct $X$ by extending $S$ in procedure (P), 
let us visit $X$ if and only if 
\begin{center}
\begin{tabular}{@{}l@{\hspace*{3mm}}p{10cm}}
(T1) & $(X,S)\cong (X,M(X))$. 
\end{tabular}
\end{center}

\begin{lemma}
\label{lem:mckay-t1}
The procedure (P) equipped with the test (T1) visits every isomorphism
class of objects in $\Omega$ at least once. Furthermore, 
any two visits to the same isomorphism class must 
(i) originate by extension from the same procedure invocation on input $S$, and 
(ii) belong to the same $\Aut(S)$-orbit of this seed $S$.
\end{lemma}

\begin{proof}
Suppose that $X$ is visited by extending $S$ and $Y$ is visited by
extending $T$, with $X\cong Y$. By (T1) we must thus have 
$(X,S)\cong(X,M(X))$ and $(Y,T)\cong(Y,M(Y))$. Furthemore, from (M2) 
we have $(X,M(X))\cong (Y,M(Y))$. Thus, $(X,S)\cong (Y,T)$ and 
hence $S\cong T$. Since $S\cong T$, we must in fact have $S=T$ by our
assumption on how procedure (P) is invoked. Since $X$ and $Y$ were arbitrary, 
any two visits
to the same isomorphism class must originate by extension from the same seed.
Furthermore, we have $(X,S)\cong (Y,S)$ and thus $X\cong_{\Aut(S)} Y$.
Let us next observe that every isomorphism class of objects is visited at least 
once. Indeed, let $Y\in\Omega$ be arbitrary. By (M1), we have $\erel{Y}{M(Y)}$. 
In particular, there is a unique $S\in\Sigma$ with $S\cong M(Y)$ such that
procedure (P) is invoked with input $S$. Let $\gamma\in\Gamma$ be an associated
isomorphism with $S^\gamma=M(Y)$. By (E1), we have $\erel{X}{S}$ for
$X=Y^{\gamma^{-1}}$. Furthermore, $X\cong Y$ implies by (M2) that 
$(X,M(X))\cong (Y,M(Y))=(X^\gamma,S^\gamma)\cong (X,S)$, so (T1) holds 
and $X$ is visited. Since $X\cong Y$ and $Y$ was arbitrary, every 
isomorphism class is visited at least once. 
\end{proof}

Let us next observe that the outcome of test (T1) is invariant on each
$\Aut(S)$-orbit of extensions of $S$.
\begin{lemma}
\label{lem:mckay-t2}
For all $\alpha\in\Aut(S)$ 
we have that 
(T1) holds for $(X,S)$ if and only if 
(T1) holds for $(X^\alpha,S)$.
\end{lemma}

\begin{proof}
From $X\cong X^\alpha$ and (M2) we have 
$(X,M(X))\cong (X^\alpha,M(X^\alpha))$. 
Thus, $(X,S)\cong (X,M(X))$ if and only if
$(X^\alpha,S)
=(X^\alpha,S^\alpha)
\cong (X,S)
\cong (X,M(X))
\cong (X^\alpha,M(X^\alpha))$.
\end{proof}

Lemma~\ref{lem:mckay-t2} in particular implies that we obtain complete
isomorph rejection by combining the test (T1) with a further test that
ensures complete isomorph rejection on $\Aut(S)$-orbits. 
Towards this end, let us associate an arbitrary order relation on every 
$\Aut(S)$-orbit on $e(S)$. Let us perform the following further test:
\begin{center}
\begin{tabular}{@{}l@{\hspace*{3mm}}p{10cm}}
(T2) & $X=\min X^{\Aut(S)}$. 
\end{tabular}
\end{center}

\noindent
The following lemma is immediate from Lemma~\ref{lem:mckay-t1} 
and Lemma~\ref{lem:mckay-t2}.
\begin{lemma}
\label{lem:mckay-full} 
The procedure (P) equipped with the tests (T1) and (T2) 
visits every isomorphism class of objects in $\Omega$ exactly once. 
\end{lemma}

\subsection{A template for canonical extension maps}
We conclude this section by describing 
a template of how to use an arbitrary canonical labeling map 
$\kappa:\Omega\rightarrow \Gamma$ to construct a canonical extension map
$M:\Omega\rightarrow\Sigma$. 

For $X\in\Omega$ construct the canonical form $Z=X^{\kappa(X)}$.
Using the canonical form $Z$ only, identify a seed 
$T$ with $\erel{Z}{T}$. In particular, such a seed must exist by (E2). 
(Typically this identification can be carried out by studying $Z$ and finding
an appropriate substructure in $Z$ that qualifies as $T$. For example, $T$
may be the minimum seed in $\Sigma$ that satisfies $\erel{Z}{T}$. 
Cf.~Lemma~\ref{lem:t1p-correctness}.) 
Once $T$ has been identified, set $M(X)=T^{\kappa(X)^{-1}}$.
\begin{lemma}
\label{lem:canonical-extension-from-canonical-labeling}
The map $X\mapsto M(X)$ above is a canonical extension map.
\end{lemma}

\begin{proof}
By (E1) we have $\erel{X}{M(X)}$ because $Z^{\kappa(X)^{-1}}=X$, 
$T^{\kappa(X)^{-1}}=M(X)$, and $\erel{Z}{T}$. 
Thus, (M1) holds for $M$. To verify (M2), let $X,Y\in\Omega$ with 
$X\cong Y$ be arbitrary. Since $X\cong Y$, by (K) we have 
$X^{\kappa(X)}=Z=Y^{\kappa(Y)}$. 
It follows that $M(X)=T^{\kappa(X)^{-1}}$ and $M(Y)=T^{\kappa(Y)^{-1}}$,
implying that $\gamma=\kappa(X)\kappa(Y)^{-1}$ is an isomorphism witnessing
$(X,M(X))\cong (Y,M(Y))$. Thus, (M2) holds for $M$.
\end{proof}

\section{Generation of partial assignments via a prefix sequence}

\label{sect:prefix-method}

This section describes an instantiation of McKay's method that
generates partial assignments of values to a set of variables $U$ one
variable at a time following a \emph{prefix sequence} at the level of
isomorphism classes given by the action of a group $\Gamma$ on $U$. We
postpone the extension to include variables to
Section~\ref{sect:value-symmetry}. Let $R$ be a finite set where the
variables in $U$ take values.

\subsection{Partial assignments, isomorphism, restriction}
\label{sect:partial}

For a subset $W\subseteq U$ of variables, let us say that 
a \emph{partial assignment} of values to $W$ is a mapping $X:W\rightarrow R$. 
Isomorphism for partial assignments is induced by the following group action. 
Let $\gamma\in \Gamma$ act
on $X:W\rightarrow R$ by setting 
$X^\gamma:W^\gamma\rightarrow R$ where 
$X^\gamma$ is defined for all $u\in W^\gamma$ by 
\begin{equation}
\label{eq:assignment-action}
X^\gamma(u)=X(u^{\gamma^{-1}})\,. 
\end{equation}
\begin{lemma}
\label{lem:action}
The action \eqref{eq:assignment-action} is well-defined.
\end{lemma}

\begin{proof}
We observe that for the identity $\epsilon\in\Gamma$ of $\Gamma$ we
have $X^\epsilon=X$. Furthermore, for all $\gamma,\beta\in \Gamma$ and 
$u\in W^{\gamma\beta}=(W^\gamma)^\beta$ we have 
\[
X^{\gamma\beta}(u)
=X\bigl(u^{(\gamma\beta)^{-1}}\bigr)
=X\bigl((u^{\beta^{-1}})^{\gamma^{-1}}\bigr)
=X^\gamma\bigl(u^{\beta^{-1}}\bigr)
=(X^\gamma)^\beta(u)\,.\qedhere
\]
\end{proof}

For an assignment $X:W\rightarrow R$, let us write $\underline{X}=W$ for
the underlying set of variables assigned by $X$. Observe that the underline
map is a homomorphism of group actions in the sense that 
\begin{equation}
\label{eq:homomorphism}
\underline{X^\gamma}=\underline{X}^\gamma 
\end{equation}
holds for all $\gamma\in \Gamma$ and $X:W\rightarrow R$.
For $Q \subseteq \underline{X}$, let us write $X|_Q$ for the
restriction of $X$ to $Q$. 

\subsection{The prefix sequence and generation of normalized assignments}
We are now ready to describe the generation procedure. Let us begin by
prescribing the prefix sequence. 
Let $u_1,u_2,\ldots,u_k$ be $k$ distinct elements of $U$ and
let $U_j=\{u_1,u_2,\ldots,u_j\}$ for $j=0,1,\ldots,k$. In particular
we observe that 
\[
U_0\subseteq U_1\subseteq\cdots\subseteq U_k
\]
with $U_j\setminus U_{j-1}=\{u_j\}$ for all $j=1,2,\ldots,k$.

For $j=0,1,\ldots,k$ let $\Omega_j$ consist of all partial assignments 
$X:W\rightarrow R$ with $W\cong U_j$. Or what is the same, using the 
underline notation, $\Omega_j$ consists of all partial assignments $X$ with 
$\underline{X}\cong U_j$.

We rely on canonical extension to construct exactly one object from each 
orbit of $\Gamma$ on $\Omega_j$, using as seeds exactly one object from each
orbit of $\Gamma$ on $\Omega_{j-1}$, for each $j=1,2,\ldots,k$.
We assume the availability of canonical labeling maps 
$\kappa:\Omega_j\rightarrow\Gamma$ for each $j=1,2,\ldots,k$.

Our construction procedure will work with objects that are in a normal form
to enable precomputation for efficient execution of the subsequent 
tests for isomorph rejection. Towards this end, 
let us say that $X\in\Omega_j$ is \emph{normalized} if $\underline{X}=U_j$. 
It is immediate from our definition of $\Omega_j$ 
and \eqref{eq:assignment-action} that each orbit in $\Omega_j/\Gamma$ 
contains at least one normalized object. 

Let us begin with a high-level description of the construction procedure,
to be followed by the details of the isomorph rejection tests and a proof
of correctness. Fix $j=1,2,\ldots,k$ and 
study the following procedure, which we assume is invoked for exactly
one normalized representative $S\in\Omega_{j-1}$ from each orbit 
in $\Omega_{j-1}/\Gamma$:
\begin{center}
\begin{tabular}{@{}l@{\hspace*{3mm}}p{10.5cm}}
(P') & 
Let a normalized $S\in\Omega_{j-1}$ be given as input. 
For each $p\in u_j^{\Aut(U_{j-1})}$ and each $r\in R$, 
construct the assignment 
\[
X:U_{j-1}\cup\{p\}\rightarrow R
\]
defined by $X(p)=r$ and $X(u)=S(u)$ for all $u\in U_{j-1}$. 
Perform the isomorph rejection tests (T1') and (T2') on $X$ and $S$.
If both tests accept, visit $X^{\nu(p)}$ where $\nu(p)\in\Aut(U_{j-1})$ 
normalizes $X$.
\end{tabular}
\end{center}

\medskip
\noindent
From an implementation point of view, 
it is convenient to precompute the orbit $u_j^{\Aut(U_{j-1})}$ 
together with group elements $\nu(p)\in\Aut(U_{j-1})$ for 
each $p\in u_j^{\Aut(U_{j-1})}$ that satisfy 
$p^{\nu(p)}=u_j$. Indeed, a constructed $X$ with 
$\underline X=U_{j-1}\cup\{p\}$ can now be normalized by acting with 
$\nu(p)$ on $X$ to obtain a normalized $X^{\nu(p)}$ isomorphic to $X$. 

\subsection{The isomorph rejection tests} 
Let us now complete the description of
procedure (P') by describing the two isomorph rejection tests (T1') and (T2').
This subsection only describes the tests with an implementation in mind, the
correctness analysis is postponed to the following subsection.

Let us assume that the elements of $U$ have been arbitrarily ordered
and that $\kappa:\Omega_j\rightarrow\Gamma$ is a canonical labeling
map. Suppose that $X$ has been constructed by extending
a normalized $S$ with
$\underline{X}=\underline{S}\cup\{p\}=U_{j-1}\cup\{p\}$.  The first
test is:
\smallskip
\begin{center}
\begin{tabular}{@{}l@{\hspace*{3mm}}p{10cm}}
(T1') & 
Subject to the ordering of $U$, select the minimum $q\in U$ 
such that $q^{\kappa(X)^{-1}\nu(p)}\in u_j^{\Aut(U_j)}$.
Accept if and only if $p\cong_{\Aut(X)} q^{\kappa(X)^{-1}}$.
\end{tabular}
\end{center}

\smallskip
\noindent
From an implementation perspective we observe that we can precompute
the orbit $u_j^{\Aut(U_j)}$. Furthermore, the only computationally nontrivial
part of the test is the computation of $\kappa(X)$ since we assume that we 
obtain generators for $\Aut(X)$ as a side-effect of this computation. Indeed,
with generators for $\Aut(X)$ available, it is easy to compute the orbits
$U/\Aut(X)$ and hence to test whether $p\cong_{\Aut(X)} q^{\kappa(X)^{-1}}$.

Let us now describe the second test:
\begin{center}
\begin{tabular}{@{}l@{\hspace*{3mm}}p{10cm}}
(T2') & 
Accept if and only if $p=\min p^{\Aut(S)}$
subject to the ordering of $U$.
\end{tabular}
\end{center}

\smallskip
\noindent
From an implementation perspective we observe that 
since $S$ is normalized we have 
$\Aut(S)\leq\Aut(\underline{S})=\Aut(U_{j-1})$ and thus the orbit 
$u_j^{\Aut(U_{j-1})}$ considered by procedure (P') partitions into 
one or more $\Aut(S)$-orbits. Furthermore, generators for $\Aut(S)$ are 
readily available (due to $S$ itself getting accepted in the test (T1') 
at an earlier level of recursion), and thus the orbits 
$u_j^{\Aut(U_{j-1})}/\Aut(S)$ and their minimum elements are cheap to compute. 
Thus, a fast implementation of procedure (P') will in most cases execute 
the test (T2') before the more expensive test (T1').

\begin{example}
\label{exa:tree}
  We display below a possible search tree for 
  the system of clauses \eqref{eq:example} and
  the prefix sequence $x_3,x_4,x_5,x_6$.
  Each node in the search tree displays the prefix assignment $X$ (top),
  its canonical version $X^{\kappa(X)}$ (middle)
  and
  its normalized version $X^{\nu(p)}$ (bottom).
  The variables have the Boolean domain $\{\False,\True\}$ and
  the assignments are given in the literal form; for example, we write
  $\bar x_3 x_4$ for the assignment $\{x_3 \mapsto \False, x_4 \mapsto \True\}$.
  \medskip
  
  \noindent\includegraphics[width=\textwidth]{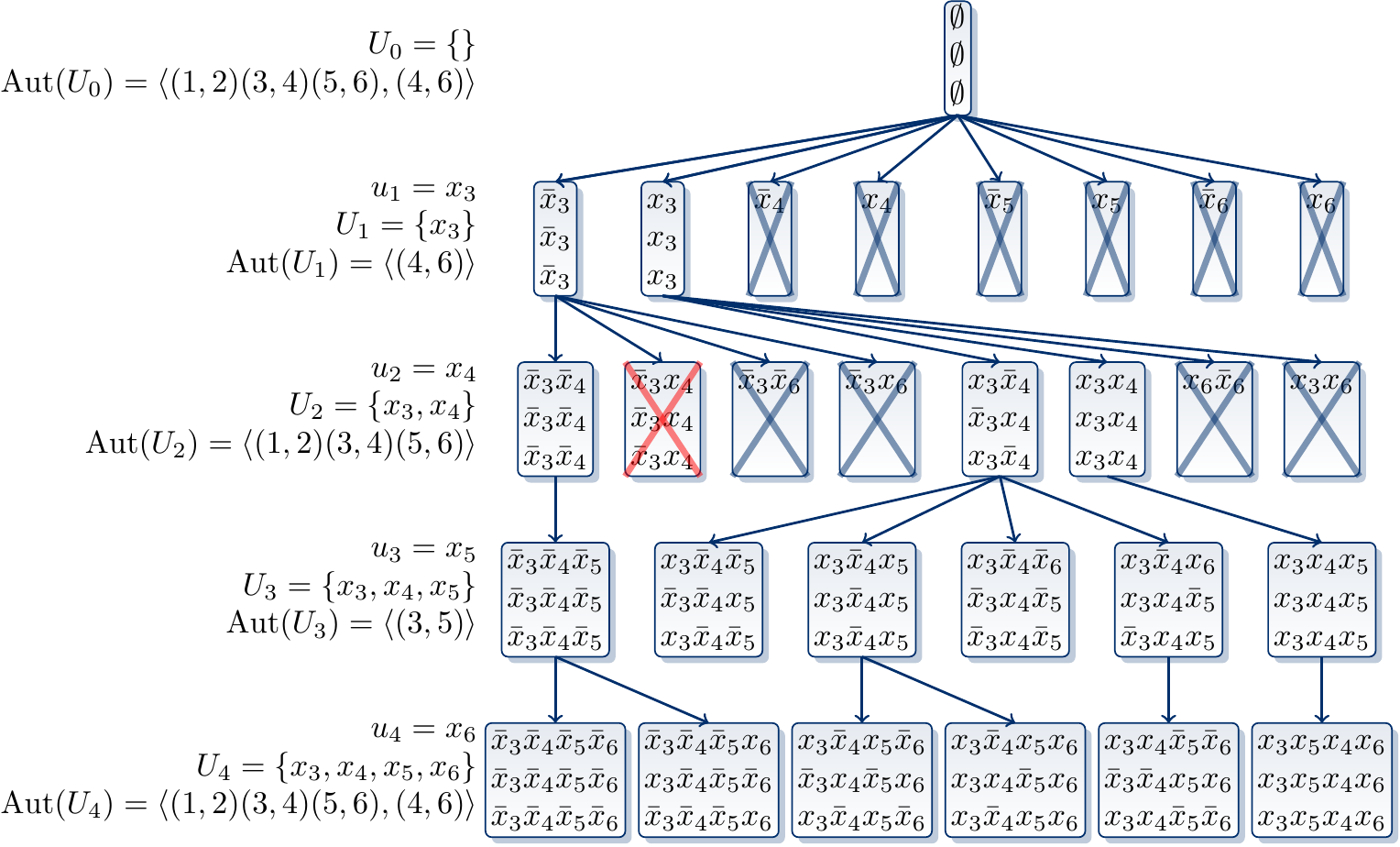}

  \medskip
  \noindent
  The nodes with a red cross are nodes eliminated by the test (T1') and
  the ones with a blue cross are eliminated by the test (T2').
  (For convenience of display, these eliminated nodes are only drawn in 
  the first three levels above.)
  For instance, the node with $X = \bar x_3 x_4$ is eliminated by the test (T1')
  because the minimum $q$ such that $q^{\kappa(X)^{-1} \nu(x_4)} \in x_4^{\Aut(U_2)}=\{x_3,x_4\}$ when $\kappa(X) = \{x_3 \mapsto x_3,x_4 \mapsto x_4\}$ and $\nu(x_4) = \{x_3 \mapsto x_3,x_4 \mapsto x_4\}$ is $x_3$ and
  $x_4 \not\cong_{\Aut(X)} q^{\kappa(X)^{-1}}=x_3$ as
  $\Aut(X) = \{\epsilon\}$.
  On the other hand,
  the node with $X = \bar x_3 \bar x_6$ is eliminated by the test (T2')
  as $x_6 \neq \min x_6^{\Aut(\bar x_3)}$ and
  $x_6^{\Aut(\bar x_3)} = \{x_4,x_6\}$.
  We observe that the search tree has dead-end nodes that do not
  extend to any full prefix assignment.
\end{example}

\subsection{Correctness}
We now establish the correctness of procedure (P') together with 
the tests (T1') and (T2') by reduction to McKay's framework 
and~Lemma~\ref{lem:mckay-full}. 
Fix $j=1,2,\ldots,k$. 
Let us start by defining the extension relation 
$e\subseteq\Omega_j\times\Omega_{j-1}$ 
for all $X\in\Omega_j$ and $S\in\Omega_{j-1}$ by 
setting $\erel{X}{S}$ if and only if 
\begin{equation}
\label{eq:extends-relation}
\text{there exists a $\gamma\in \Gamma$ such that 
$\underline{X}^\gamma=U_j$, 
$\underline{S}^\gamma=U_{j-1}$, and $X^\gamma|_{U_{j-1}}=S^\gamma$}.
\end{equation}
This relation is well-defined in the context of McKay's framework:

\begin{lemma}
\label{lem:extends}
The relation \eqref{eq:extends-relation} satisfies (E1) and (E2).
\end{lemma}

\begin{proof}
To establish (E1), let $X\in\Omega_j$ and $S\in\Omega_{j-1}$ be arbitrary.
It suffices to show that for all $\beta\in\Gamma$
we have $\erel{X}{S}$ if and only if $\erel{X^\beta}{S^\beta}$.
Let $\beta\in\Gamma$ be arbitrary. By~\eqref{eq:homomorphism}, for all
$\gamma\in\Gamma$ we have $\underline{X}^\gamma=U_j$ 
if and only if 
$\underline{X^\beta}\phantom{)\!\!}^{\beta^{-1}\gamma}=\underline{X}^{\beta\beta^{-1}\gamma}
=\underline{X}^\gamma=U_j$. Similarly, 
for any
$\gamma\in\Gamma$ we have $\underline{S}^\gamma=U_{j-1}$ 
if and only if 
$\underline{S^\beta}\phantom{)\!\!}^{\beta^{-1}\gamma}=\underline{S}^{\beta\beta^{-1}\gamma}
=\underline{S}^\gamma=U_{j-1}$. 
Finally, for any $\gamma\in\Gamma$ that satisfies 
$\underline{X}^\gamma=U_j$ and $\underline{S}^\gamma=U_{j-1}$
(equivalently, $\beta^{-1}\gamma$ satisfies
$\underline{X^\beta}\phantom{)\!\!}^{\beta^{-1}\gamma}=U_j$ and
$\underline{S^\beta}\phantom{)\!\!}^{\beta^{-1}\gamma}=U_{j-1}$),
we have 
$X^\gamma|_{U_{j-1}}=S^\gamma$
if and only if 
$(X^\beta)^{\beta^{-1}\gamma}|_{U_{j-1}}=
X^\gamma|_{U_{j-1}}=
S^\gamma=
(S^\beta)^{\beta^{-1}\gamma}$.
To establish (E2), observe that 
for an arbitrary $X\in\Omega_j$ there exists a $\gamma\in \Gamma$ with 
$\underline{X}^\gamma=U_j$, and 
thus $\erel{X}{S}$ holds for $S=T^{\gamma^{-1}}$, 
where $T$ is obtained from $Y=X^\gamma$ by deleting the assignment to 
the variable $u_j$. 
\end{proof}

The following lemma establishes that the iteration in 
procedure~(P') constructs exactly the objects $X\in e(S)$; cf.~procedure~(P).

\begin{lemma}
\label{lem:extends-set}
Let $S\in\Omega_{j-1}$ be normalized. For all $X\in\Omega_j$ we have
$\erel{X}{S}$ if and only if there exists a 
$p\in u_j^{\Aut(U_{j-1})}$ with 
$\underline{X}=U_{j-1}\cup\{p\}$ and $X|_{U_{j-1}}=S$.
\end{lemma}

\begin{proof}
From \eqref{eq:extends-relation} we have that $\erel{X}{S}$ if and only if 
there exists a $\gamma\in \Gamma$ with $\underline{X}^\gamma=U_j$, 
$\underline{S}^{\gamma}=U_{j-1}$,
and $X^\gamma|_{U_{j-1}}=S^\gamma$. 
Since $S$ is normalized,
we have $\underline{S}=U_{j-1}$ and hence
$U_{j-1}^\gamma=\underline{S}^{\gamma}=U_{j-1}$. Thus, 
$\gamma\in \Aut(U_{j-1})$ and 
\begin{equation}
\label{eq:act-restrict}
X|_{U_{j-1}}
=X^{\gamma\gamma^{-1}}|_{U_{j-1}}
=(X^\gamma|_{U_{j-1}})^{\gamma^{-1}}
=(S^{\gamma})^{\gamma^{-1}}
=S\,. 
\end{equation}
Thus, to establish the ``only if'' direction of the lemma, take 
$p=u_j^{\gamma^{-1}}$, and for the ``if'' direction, 
take $\gamma\in \Aut(U_{j-1})$ with $p^\gamma=u_j$.
\end{proof}

Next we show the correctness of the test (T1') by establishing that it
is equivalent with the test (T1) for a specific canonical extension function
$M$. Towards this end, let us use the assumed canonical labeling map 
$\kappa:\Omega_j\to\Gamma$ to build a canonical extension 
function $M$ using the template of 
Lemma~\ref{lem:canonical-extension-from-canonical-labeling}. In particular,
given an $X\in\Omega_j$ as input with $\underline{X}=U_{j-1}\cup\{p\}$, 
first construct the canonical form $Z=X^{\kappa(X)}$. In accordance 
with (T1'), select the minimum $q\in U$ 
such that $q^{\kappa(X)^{-1}\nu(p)}\in u_j^{\Aut(U_{j})}$. 
Now construct $M(X)$ from $X$ by 
deleting the value of $q^{\kappa(X)^{-1}}$.

\begin{lemma}
\label{lem:t1p-correctness}
The mapping $X\mapsto M(X)$ is well-defined and satisfies
both (M1) and (M2). 
\end{lemma}

\begin{proof}
From \eqref{eq:homomorphism} we have both 
$\Aut(Z)\leq\Aut(\underline{Z})$ and 
$\underline{Z}^{\kappa(X)^{-1}\nu(p)}=U_j$. Thus,
\[
\Aut(Z)^{\kappa(X)^{-1}\nu(p)}
\leq\Aut(\underline{Z})^{\kappa(X)^{-1}\nu(p)}
=\Aut(U_j)\,.
\]
It follows that the choice of $q$ depends on $Z$ and $u_j$ but not on the 
choices of $\kappa(X)$ or $\nu(p)$. Furthermore, we observe that 
$q\in\underline{Z}$ and $q^{\kappa(X)^{-1}}\in\underline{X}$.
Thus, the construction of $M(X)$ is well-defined and (M2) holds by 
Lemma~\ref{lem:canonical-extension-from-canonical-labeling}. 

To verify (M1), observe that since
$q^{\kappa(X)^{-1}\nu(p)}\in u_j^{\Aut(U_{j})}$, there exists
an $\alpha\in\Aut(U_{j})$ with $q^{\kappa(X)^{-1}\nu(p)\alpha}=u_j$.
Thus, for $\gamma=\nu(p)\alpha$ we have
$\underline{X}^\gamma
=(U_{j-1}\cup\{p\})^{\nu(p)\alpha}
=U_j^\alpha=U_j$,
$\underline{M(X)}^\gamma
=(U_j\setminus\{q^{\kappa(X)^{-1}\nu(p)}\})^\alpha
=U_{j-1}$,
and 
$X^\gamma|_{U_{j-1}}=M(X)^{\gamma}$. 
Thus, from \eqref{eq:extends-relation} we have $\erel{X}{M(X)}$ and thus (M1) holds.
\end{proof}

To complete the equivalence between (T1') and (T1), 
observe that since $X$ and $p$ determine $S$ 
by $X|_{\underline{X}\setminus\{p\}}=S$, 
and similarly $X$ and $q^{\kappa(X)^{-1}}$
determine $M(X)$ by $X|_{\underline{X}\setminus\{q^{\kappa(X)^{-1}}\}}=M(X)$,
the test (T1) is equivalent to testing whether 
$(X,p)\cong(X,q^{\kappa(X)^{-1}})$ holds, that is, whether
$p\cong_{\Aut(X)} q^{\kappa(X)^{-1}}$ holds. Observe that this is exactly
the test (T1'). 

It remains to establish the equivalence of (T2') and (T2). We start with
a lemma that captures the $\Aut(S)$-orbits considered by (T2).  
\begin{lemma}
\label{lem:orbit-correspondence}
For a normalized $S\in\Omega_{j-1}$ 
the orbits in $e(S)/\Aut(S)$ are in a one-to-one
correspondence with the elements of 
$(u_j^{\Aut(U_{j-1})}/\Aut(S))\times R$.
\end{lemma}

\begin{proof}
From \eqref{eq:assignment-action} we have 
$\Aut(S)\leq\Aut(\underline{S})=\Aut(U_{j-1})$ since $S$ is normalized.
Furthermore, Lemma~\ref{lem:extends-set}
implies that every extension $X\in e(S)$ is uniquely determined by 
the variable $p\in u_j^{\Aut(U_{j-1})}\cap\underline{X}$ and the 
value $X(p)\in R$. Since the action~\eqref{eq:assignment-action} 
fixes the values in $R$ elementwise, for any $X,X'\in e(S)$ 
we have $X\cong_{\Aut(S)}X'$ if and only if both $p\cong_{\Aut(S)} p'$ 
and $X(p)=X'(p')$. The lemma follows.
\end{proof}

Now order the elements $X\in e(S)$ based on the lexicographic
ordering of the pairs $(p,X(p))\in u_j^{\Aut(U_{j-1})}\times R$. 
Since the action \eqref{eq:assignment-action} 
fixes the values in $R$ elementwise, we have that (T2') holds 
if and only if (T2) holds for this ordering of $e(S)$.
The correctness of procedure (P') equipped with
the tests (T1') and (T2') now follows from Lemma~\ref{lem:mckay-full}.

\subsection{Selecting a prefix} This section gives a brief discussion
on how to select the prefix.
Let $U_k=\{u_1,u_2,\ldots,u_k\}$ be the set of variables in the 
prefix sequence. It is immediate that there exist $|R|^k$ distinct 
partial assignments from $U_k$ to $R$. Let us write $R^{U_k}$ for the set of
these assignments. The group $\Gamma$ now partitions $R^{U_k}$ 
into orbits via the action \eqref{eq:assignment-action},
and it suffices to consider at most one representative from 
each orbit to obtain an exhaustive traversal of the search space, up to 
isomorphism. Our goal is thus to select the prefix $U_k$ so that the setwise 
stabilizer $\Gamma_{U_k}$ has comparatively few orbits on $R^{U_k}$ 
compared with the total number of such assignments. 
In particular, the ratio of the number of orbits $|R^{U_k}/\Gamma_{U_k}|$
to the total number of mappings $|R|^k$ can be viewed as a proxy for 
the achieved symmetry reduction and as a rough%
\footnote{%
Here it should be noted that executing the symmetry reduction carries in 
itself a nontrivial computational cost. That is, there is a tradeoff between
the potential savings in solving the system gained by symmetry reduction 
versus the cost of performing symmetry reduction. For example, if the 
instance has no symmetry and $\Gamma$ is a trivial group, then executing 
symmetry reduction merely makes it more costly to solve the system.}{}
proxy for the speedup factor obtained compared with no symmetry reduction 
at all. 

\subsection{Subroutines}
By our assumption, the canonical labeling map $\kappa$ produces as a
side-effect a set of generators for the automorphism group $\Aut(X)$
for a given input $X$. We also assume that generators for the groups
$\Aut(U_j)$ for $j=0,1,\ldots,k$ can be precomputed by similar means.
This makes the canonical labeling map essentially the only nontrivial
subroutine needed to implement procedure (P'). Indeed, the orbit
computations required by tests (T1') and (T2') are implementable by
elementary permutation group
algorithms~\cite{Butler:1991,Seress:2003}. Section~\ref{sect:colored-graphs}
describes how to implement $\kappa$ by reduction to vertex-colored
graphs.%
\footnote{Reduction to vertex-colored graphs is by no means the only possibility to obtain the canonical labeling map to enable (P'), (T1'), and (T2'). Another possibility would be to represent $\Gamma$ directly as a permutation group and use dedicated permutation-group algorithms~\cite{Leon:1991,Leon:1997}. Our present choice of vertex-colored graphs is motivated by easy availability of carefully engineered implementations for working with vertex-colored graphs.}{}

\section{Value symmetries}
\label{sect:value-symmetry}

The previous section considered prefix-assignment generation subject to
an action of a group $\Gamma$ on the set of variables $U$. In this section,
we extend the framework so that it captures symmetries in values
assigned to variables, or {\em value symmetries}. Towards this end, we 
extend the domain that records the symmetries from $U$ to $U\times R$, 
where $R$ is the set of values that can be assigned to the variables in $U$.
Accordingly, in what follows we assume that the group $\Gamma$ acts 
on $U\times R$ as well as on $U$, the latter by restriction. 

The action of the group $\Gamma$ on $U\times R$ may not be completely 
arbitrary, however, because we want partial assignments $X:W\rightarrow R$
with $W\subseteq U$ to remain well-defined functions under the action 
of $\Gamma$. This property is naturally captured by 
the \emph{wreath product} group $\Sym(R)\wr \Sym(U)$ and its natural action 
on $U\times R$. 

\subsection{The wreath product and its actions}

We will follow the convention that $\Sym(R)\wr \Sym(U)$
acts on $U\times R$ by first acting on $U$ and then on $R$. 
For accessibility and convenience, we review our conventions in detail. 
The group $\Sym(R)\wr\Sym(U)$ consists of all pairs $(\pi,\sigma)$, 
where $\pi\in\Sym(U)$ is a permutation of $U$ and 
$\sigma:U\rightarrow\Sym(R)$ associates a permutation $\sigma(u)\in\Sym(R)$ 
with each element $u\in U$.
In particular, $\Sym(R)\wr\Sym(U)$ has order $|U|!\cdot(|R|!)^{|U|}$.

The product of two elements
$(\pi_1,\sigma_1),(\pi_2,\sigma_2)\in\Sym(R)\wr\Sym(U)$
is defined by $(\pi,\sigma)=(\pi_1,\sigma_1)(\pi_2,\sigma_2)$, where
\begin{equation}
\label{eq:wr-operation-pi}
\pi=\pi_1\pi_2
\end{equation}
and for all $u\in U$ we set
\begin{equation}
\label{eq:wr-operation-sigma}
\sigma(u)=\sigma_1(u^{\pi_2^{-1}})\sigma_2(u)\,.
\end{equation}
The inverse of an element $(\pi,\sigma)\in\Sym(R)\wr\Sym(U)$ is thus given by
$(\pi,\sigma)^{-1}=(\rho,\tau)$, where 
\begin{equation}
\label{eq:wr-inverse-pi}
\rho=\pi^{-1}
\end{equation}
and for all $u\in U$ we have
\begin{equation}
\label{eq:wr-inverse-sigma}
\tau(u)=\sigma(u^{\pi})^{-1}\,.
\end{equation}

An element $(\pi,\sigma)\in\Sym(R)\wr\Sym(U)$ acts on an element
$u\in U$ by
\begin{equation}
\label{eq:wr-action-u}
u^{(\pi,\sigma)}=u^\pi
\end{equation}
and on a pair $(u,r)\in U\times R$ by
\begin{equation}
\label{eq:wr-action-ur}
(u,r)^{(\pi,\sigma)}=(u^\pi,r^{\sigma(u^{\pi})})\,.
\end{equation}
Here in particular the intuition is that we first act on $(u,r)$ with
$\pi$ to obtain $(u^\pi,r)$, and then act with $\sigma(u^{\pi})$ 
to obtain $(u^\pi,r^{\sigma(u^{\pi})})$. Extend the 
action \eqref{eq:wr-action-u} elementwise to subsets $W\subseteq U$.

\subsection{Partial assignments and isomorphism}

Let $\Gamma$ be a subgroup of $\Sym(R)\wr\Sym(U)$ and let $\Gamma$ act on $U$ 
and $U\times R$ by \eqref{eq:wr-action-u} and \eqref{eq:wr-action-ur},
respectively. Furthermore, we let an element $\gamma=(\pi,\sigma)\in\Gamma$ 
act on a partial assignment $X:W\rightarrow R$ with $W\subseteq U$ to 
produce the partial assignment $X^\gamma:W^\pi\rightarrow R$ 
defined for all $u\in W^\pi$ by
\begin{equation}
\label{eq:value-assignment-action}
X^\gamma(u) = X(u^{\pi^{-1}})^{\sigma(u)}\,.
\end{equation}

In analogy with Lemma~\ref{lem:action}, let us verify that the value-permuting
action~\eqref{eq:value-assignment-action} is well-defined.
\begin{lemma}
The action~\eqref{eq:value-assignment-action} is well defined.
\end{lemma}
\begin{proof}
We observe that for the identity 
$\epsilon\in\Gamma$ of $\Gamma\leq\Sym(R)\wr\Sym(U)$, we
have $X^\epsilon=X$. 
Furthermore, for all $\gamma_1=(\pi_1,\sigma_1)\in\Gamma$, $\gamma_2=(\pi_2,\sigma_2)\in \Gamma$, and 
$u\in W^{\gamma_1\gamma_2}=(W^{\gamma_1})^{\gamma_2}$, 
by \eqref{eq:value-assignment-action}, 
\eqref{eq:wr-operation-pi}, and
\eqref{eq:wr-operation-sigma}, we have
\begin{align*}
X^{\gamma_1\gamma_2}(u) 
& = X(u^{(\pi_1\pi_2)^{-1}})^{\sigma_1(u^{\pi_2^{-1}})\sigma_2(u)}
  = X(u^{\pi_2^{-1}\pi_1^{-1}})^{\sigma_1(u^{\pi_2^{-1}})\sigma_2(u)} \\
& = X^{\gamma_1}(u^{\pi_2^{-1}})^{\sigma_2(u)}
  = (X^{\gamma_1})^{\gamma_2}(u) \,.\qedhere
\end{align*}
\end{proof}

Let us recall that for $X:W\rightarrow R$ we write $\underline{X}=W$ for
the underlying set of variables assigned by $X$. In analogy
with Section~\ref{sect:partial}, the underline map is a homomorphism
of group actions that satisfies \eqref{eq:homomorphism} for the action
\eqref{eq:value-assignment-action} and the action \eqref{eq:wr-action-u}
extended elementwise to subsets of $U$.
Isomorphism for partial assignments is now induced by the 
action~\eqref{eq:value-assignment-action}.

\subsection{Generating normalized assignments}

Working with the group action~\eqref{eq:value-assignment-action}, 
let $u_1,u_2,\ldots,u_k$ be $k$ distinct elements of $U$, and let
$U_j=\{u_1,u_2,\ldots,u_j\}$ for $j=0,1,\ldots,k$. Let $\Omega_j$
consist of all partial assignments $X:W\rightarrow R$ with
$W\cong U_j$. We construct exactly one object form each orbit of
$\Gamma$ on $\Omega_j$, using as seeds exactly one object from each
orbit of $\Gamma$ on $\Omega_{j-1}$, for each $j=1,2,\ldots,k$,
assuming the availability of canonical labeling maps
$\kappa:\Omega_j\rightarrow\Gamma$. We say the assignment
$X\in\Omega_j$ is {\em normalized} if $\underline{X}=U_j$.

We now present a version of the procedure (P') modified for the
group action~\eqref{eq:value-assignment-action}.
Let us fix $j=1,2,\ldots,k$. We assume that
the procedure is invoked for exactly one normalized representative
$S\in\Omega_{j-1}$ from each orbit in $\Omega_{j-1}/\Gamma$.
\NamedIndent{(P'')}{ Let a normalized $S\in\Omega_{j-1}$ be given as
  input.  For each 
  $p\in u_j^{\Aut(U_{j-1})}$ and each $r\in R$, construct the
  assignment $X:U_{j-1}\cup\{p\}\rightarrow R$ defined by $X(p)=r$ and
  $X(u)=S(u)$ for all $u\in U_{j-1}$. Perform the isomorph rejection
  tests (T1') and (T2'') on $X$ and $S$. If both tests accept, visit
  $X^{\nu(p)}$ where $\nu(p)\in\Aut(U_{j-1})$ normalizes $X$.  
}
In particular, procedure (P'') has two differences compared with
procedure (P'). First, the underlying group action 
is \eqref{eq:value-assignment-action}. Second, the test (T2') has been
replaced with a new test (T2'') to account for more extensive orbits 
of pairs $(p,r)$ under the action of $\Aut(S)$.

\subsection{The isomorph rejection tests}

Assume that the elements of $U$, $R$, and $U\times R$ have been arbitrarily 
ordered and that $\kappa:\Omega_j\rightarrow\Gamma$ is a canonical labeling
map.  Suppose that $X$ has been constructed by extending a normalized
$S$ with $\underline{X}=\underline{S}\cup\{p\}=U_{j-1}\cup\{p\}$
and $X(p)=r$. 
Let us first recall the test (T1') for convenience:
\NamedIndent{(T1')}{
Subject to the ordering of $U$, select the minimum $q\in U$ 
such that $q^{\kappa(X)^{-1}\nu(p)}\in u_j^{\Aut(U_j)}$.
Accept if and only if $p\cong_{\Aut(X)} q^{\kappa^{-1}(X)}$.
}
The new isomorph rejection test is as follows:
\NamedIndent{(T2'')}{
Accept if and only if $(p,r)=\min \, (p,r)^{\Aut(S)}$ subject to
the ordering of $U\times R$.
}

\subsection{Correctness}

We now establish the correctness of the modified procedure (P''). Fix
$j=1,2,\ldots,k$. Define the extension relation
$e\subseteq\Omega_j\times\Omega_{j-1}$ as in
\eqref{eq:extends-relation}.  This relation is well-defined in the
context of McKay's framework under the modified group action.
\begin{lemma}
\label{lem:extends-val}
The relation \eqref{eq:extends-relation} satisfies (E1) and (E2) when
the group action is as defined
in~\eqref{eq:value-assignment-action}.
\end{lemma}
\begin{proof}
Identical to Lemma~\ref{lem:extends} since \eqref{eq:homomorphism} holds
for the action \eqref{eq:value-assignment-action}
and the action \eqref{eq:wr-action-u} extended elementwise to subsets of $U$.
\end{proof}

The correctness analysis of the test (T1') proceeds identically as in 
Section~\ref{sect:partial}, relying on \eqref{eq:homomorphism} 
in the proof of Lemma~\ref{lem:t1p-correctness}.
To establish the correctness of the new test (T2''), we first observe
that the counterpart of Lemma~\ref{lem:extends-set} holds for the
modified group action.

\begin{lemma}
\label{lem:extends-set-val}
Let $S\in\Omega_{j-1}$ be normalized. For all $X\in\Omega_j$, we have
$\erel{X}{S}$ if and only if there exists a 
$p\in u_j^{\Aut(U_{j-1})}$ with 
$\underline{X}=U_{j-1}\cup\{p\}$ and $X|_{U_{j-1}}=S$.
\end{lemma}

\begin{proof}
First observe that \eqref{eq:act-restrict} holds 
for the action \eqref{eq:value-assignment-action}. Then
proceed as in the proof of Lemma~\ref{lem:extends-set}.
\end{proof}

Let us now proceed to the counterpart of Lemma~\ref{lem:orbit-correspondence}.

\begin{lemma}
\label{lem:orbit-correspondence-val}
For a normalized $S\in\Omega_{j-1}$,
the orbits in $e(S)/\Aut(S)$ are in a one-to-one
correspondence with the orbits in 
$(u_j^{\Aut(U_{j-1})}\times R)/\Aut(S)$.
\end{lemma}
\begin{proof}
Lemma~\ref{lem:extends-set-val} implies that every extension $X\in e(S)$
is uniquely determined by the variable 
$p\in u_j^{\Aut(U_{j-1})}\cap\underline{X}$ and the value $X(p)\in R$.
That is, the elements in $e(S)$ are in one-to-one correspondence with 
elements in $u_j^{\Aut(U_{j-1})}\times R$. 

Let $\Aut(S)$ act on $e(S)$ via \eqref{eq:value-assignment-action}; this 
action is well-defined by Lemma~\ref{lem:extends-val} and (E1) since 
for all $\alpha\in\Aut(S)$ we have $\erel{X}{S}$ if and only 
if $\erel{X^\alpha}{S}$. Let $\Aut(S)$ act on $u_j^{\Aut(U_{j-1})}\times R$
via \eqref{eq:wr-action-ur}; this action is well-defined because
$S$ is normalized and hence 
$\Aut(S)\leq\Aut(\underline{S})=\Aut(U_{j-1})$ holds by \eqref{eq:homomorphism}.

Let $X,Y\in e(S)$ be arbitrary 
with $\underline X = U_{j-1} \cup \{p\}$ and 
$\underline Y = U_{j-1} \cup \{q\}$.
We now claim that $X\cong_{\Aut(S)} Y$ holds under the action
\eqref{eq:value-assignment-action}
if and only if
$(p,X(p)) \cong_{\Aut(S)} (q,Y(q))$ holds under the action 
\eqref{eq:wr-action-ur}.
To see this, first observe that for all $\alpha\in\Aut(S)$ we have 
$U_{j-1}^\alpha=\underline{S}^\alpha=\underline{S}=U_{j-1}$ by
\eqref{eq:homomorphism} since $S$ is normalized. 
Furthermore, $X|_{U_{j-1}}=Y|_{U_{j-1}}=S$. 
Thus, by~\eqref{eq:value-assignment-action} it holds that 
for all $\alpha=(\pi,\sigma)\in\Aut(S)$ with $\pi\in\Sym(U)$
and $\sigma:U\rightarrow\Sym(R)$ we have $Y=X^\alpha$ if and only 
if $q=p^\pi$ and 
$Y(q)=X^\alpha(q)=X(q^{\pi^{-1}})^{\sigma(q)}=X(p)^{\sigma(p^{\pi})}$. 
Or what is the same by \eqref{eq:wr-action-ur}, if and only if 
$(q,Y(q))=(p^\pi,X(p)^{\sigma(p^\pi)})=(p,X(p))^{\alpha}$.
\end{proof}
Order the elements $X\in e(S)$ based on the lexicographic
ordering of the pairs $(p,X(p))\in u_j^{\Aut(U_{j-1})}\times R$.
We now have that (T2'') holds 
if and only if (T2) holds for this ordering of $e(S)$.
The correctness of procedure (P'') equipped with
the tests (T1') and (T2'') now follows from Lemma~\ref{lem:mckay-full}.

\section{Representation using vertex-colored graphs}
\label{sect:colored-graphs}

This section describes one possible approach to represent the group of 
symmetries $\Gamma\leq\Sym(U)$ of a system of constraints over a finite 
set of variables $U$ taking values in a finite set $R$.
Our representation of choice will be vertex-colored graphs over a fixed 
finite set of vertices $V$. In particular, isomorphisms between such graphs 
are permutations $\gamma\in\Sym(V)$ that map edges onto edges and respect the
colors of the vertices; that is, every vertex in $V$ maps to a vertex of the
same color under $\gamma$.
It will be convenient to develop the relevant graph representations in steps,
starting with the representation of the constraint system and then proceeding
to the representation of setwise stabilizers and partial assignments. These 
representations are folklore (see e.g.~\cite{KaskiOstergard:2006}) and are 
presented here for completeness of exposition only. 

\subsection{Representing the constraint system}
To capture $\Gamma\cong\Aut(G)$ via a vertex-colored graph $G$ with 
vertex set $V$, it is convenient to represent the variables $U$ directly 
as a subset of vertices $U\subseteq V$ such that no vertex in $V\setminus U$ 
has a color that agrees with a color of a vertex in $U$. We then seek a 
graph $G$ such that $\Aut(G)\leq\Sym(U)\times\Sym(V\setminus U)$ projected
to $U$ is exactly $\Gamma$. In most cases such a graph $G$ is concisely 
obtainable by encoding the system of constraints with additional vertices 
and edges joined to the vertices representing the variables in $U$. 
We discuss two examples. 

\begin{example}
\label{exa:cnf}
Consider the system of clauses \eqref{eq:example} and its graph 
representation \eqref{eq:example-graph}. The latter can be obtained as 
follows. First, introduce a blue vertex for each of the six variables 
of \eqref{eq:example}. These blue vertices constitute the subset $U$. 
Then, to accommodate negative literals, introduce a red vertex joined by 
an edge to the corresponding blue vertex representing the positive literal. 
These edges between red and blue vertices ensure that positive and negative 
literals remain consistent under isomorphism. Finally, introduce a green
vertex for each clause of \eqref{eq:example} with edges joining the clause 
with each of its literals. It is immediate that we can reconstruct
\eqref{eq:example} from \eqref{eq:example-graph} up to labeling of 
the variables even after arbitrary color-preserving permutation of the 
vertices of \eqref{eq:example-graph}. Thus, \eqref{eq:example-graph} 
represents the symmetries of \eqref{eq:example}.
\end{example}

Let us next discuss an example where it is convenient to represent the 
symmetry at the level of original constraints rather than at the level 
of clauses. 

\begin{example}
\label{exa:brent}
Consider the following system of eight cubic equations over 24 variables 
taking values modulo 2: 
\[
\begin{array}{l@{\hspace{5mm}}r}
  x_{11}y_{11}z_{11} + x_{12}y_{12}z_{12} + x_{13}y_{13}z_{13} = 0
  &
  x_{21}y_{11}z_{11} + x_{22}y_{12}z_{12} + x_{23}y_{13}z_{13} = 0
  \\
  x_{11}y_{11}z_{21} + x_{12}y_{12}z_{22} + x_{13}y_{13}z_{23} = 0
  &
  x_{21}y_{11}z_{21} + x_{22}y_{12}z_{22} + x_{23}y_{13}z_{23} = 1
  \\
  x_{11}y_{21}z_{11} + x_{12}y_{22}z_{12} + x_{13}y_{23}z_{13} = 1
  &
  x_{21}y_{21}z_{11} + x_{22}y_{22}z_{12} + x_{23}y_{23}z_{13} = 1
  \\
  x_{11}y_{21}z_{21} + x_{12}y_{22}z_{22} + x_{13}y_{23}z_{23} = 1
  &
  x_{21}y_{21}z_{21} + x_{22}y_{22}z_{22} + x_{23}y_{23}z_{23} = 1 
\end{array}
\]
This system seeks to decompose a $2 \times 2 \times 2$ 
tensor (whose elements appear on the right hand sides of the equations)
into a sum of three rank-one tensors. The symmetries of addition and multiplication modulo 2 imply that the symmetries of the system can be 
represented by the following vertex-colored graph:

\begin{center}
\includegraphics[width=0.99\textwidth]{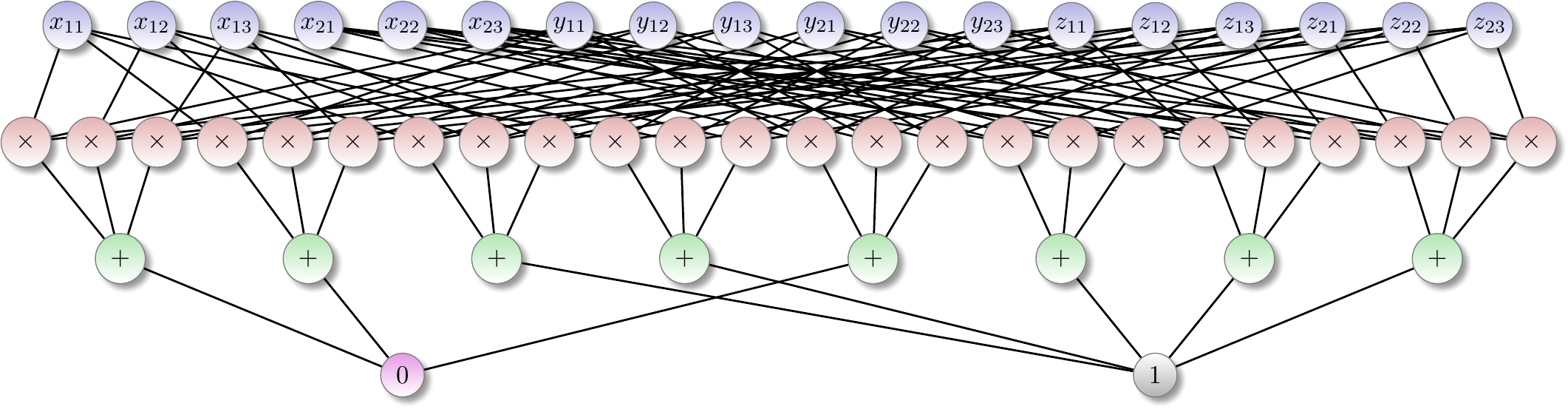}
\end{center}

\noindent
Indeed, we encode each monomial in the system with a product-vertex, and 
group these product-vertices together by adjacency to a sum-vertex to 
represent each equation, taking care to introduce two uniquely colored 
constant-vertices to represent the right-hand side of each equation.
\end{example}

\noindent
\emph{Remark.}
The representation built directly from the system of polynomial equations 
in Example~\ref{exa:brent} concisely captures the symmetries in the system 
independently of the final encoding of the system (e.g. as CNF) for solving 
purposes. 
In particular, building the graph representation from such a final 
CNF encoding (cf.~Example~\ref{exa:cnf}) results in a less compact graph
representation and obfuscates the symmetries of the original system, 
implying less efficient symmetry reduction.

\subsection{Representing the values}
In what follows it will be convenient to assume that the graph $G$ 
contains a uniquely colored vertex for each value in $R$. 
(Cf. the graph in Example~\ref{exa:brent}.)
That is, we assume that $R\subseteq V\setminus U$ and that $\Aut(G)$ 
projected to $R$ is the trivial group.

\subsection{Representing setwise stabilizers in the prefix chain}
To enable procedure (P') and the tests (T1') and (T2'), we require
generators for $\Aut(U_j)\leq\Gamma$ for each $j=0,1,\ldots,k$. 
More generally, given a subset $W\subseteq U$, we seek to compute a set of 
generators for the setwise stabilizer 
$\Aut_\Gamma(W)=\Gamma_W=\{\gamma\in\Gamma:W^\gamma=W\}$, 
with $W^\gamma=\{w^\gamma:w\in W\}$.
Assuming we have available a vertex-colored graph $G$ that represents $\Gamma$ 
by projection of $\Aut_{\Sym(V)}(G)$ to $U$, let us define the graph 
$G\!\uparrow\!W$ by selecting one vertex $r\in R$ and joining each vertex 
$w\in W$ with an edge to the vertex $r$. It is immediate that 
$\Aut_{\Sym(V)}(G\!\uparrow\!W)$ projected to $U$ is precisely 
$\Aut_\Gamma(W)$. 

\subsection{Representing partial assignments}
Let $X:W\rightarrow R$ be an assignment of values in $R$ to variables
in $W\subseteq U$. Again to enable procedure (P') together with 
the tests (T1') and (T2'), we require a canonical labeling $\kappa(X)$ and
generators for the automorphism group $\Aut(X)$. 
Again assuming we have a vertex-colored graph $G$ 
that represents $\Gamma$, let us define the graph $G\!\uparrow\!X$ by
joining each vertex $w\in W$ with an edge to the vertex $X(w)\in R$. 
It is immediate that $\Aut_{\Sym(V)}(G\!\uparrow\!X)$ projected to $U$ is 
precisely $\Aut_\Gamma(X)$. Furthermore, a canonical labeling $\kappa(X)$
can be recovered from $\kappa(G\!\uparrow\!X)$ and the 
canonical form $(G\!\uparrow\!X)^{\kappa(G\uparrow\!X)}$.

\subsection{Using tools for vertex-colored graphs}
Given a vertex-colored graph $G$ as input, practical tools exist for computing
a canonical labeling $\kappa(G)\in\Sym(V)$ and a set of generators for
$\Aut(G)\leq\Sym(V)$. Such tools include 
\emph{bliss}~\cite{JunttilaKaski:2007},
\emph{nauty}~\cite{McKay:1981,McKayPiperno:2014}, 
and \emph{traces}~\cite{McKayPiperno:2014}. Once the canonical labeling
and generators are available in $\Sym(V)$ it is easy to map back 
to $\Gamma$ by projection to $U$ so that corresponding elements 
of $\Gamma$ are obtained. 

\section{Parallel implementation}
\label{sect:implementation}

This section outlines the parallel implementation of our technique into
a tool called \texttt{reduce}. The implementation is written in C++ and 
structured as a preprocessor that works with an explicitly given graph 
representation. In the absence of such an input graph, the graph is 
constructed automatically from CNF as described in 
Section~\ref{sect:colored-graphs}. The
\emph{nauty}~\cite{McKay:1981,McKayPiperno:2014}{} canonical labeling
software for vertex-colored graphs is utilized as a subroutine.

\subsection{Backtracking search for partial assignments} 

The backtracking search for partial assignments is implemented using 
a stack that stores nodes of the search tree. 
(Recall Example~\ref{exa:tree} for an illustration of a search tree.) 
Each node $X$ in the stack represents the 
complete subtree of the search tree rooted at $X$. Initially, the stack 
consists of the empty assignment, which represents the entire search tree. 

Throughout the search, we maintain the invariant that the nodes stored in the
stack represent pairwise node-disjoint subtrees of the search tree, which 
enables us to work through the contents of the stack in arbitrary order and 
to distribute the contents of the stack to multiple compute nodes as necessary; 
we postpone a detailed discussion of the distribution of the stack and 
parallelization of the search to Section~\ref{sect:parallelization}.

Viewed as a sequential process, 
the search proceeds by iterating the following work procedure until the
stack is empty:

\begin{center}
\begin{tabular}{@{}l@{\hspace*{3mm}}p{10.5cm}}
(W) & 
Pop an assignment $X_\ell$ from the stack. Unless $X_\ell$ 
is the empty assignment (that is, unless $\ell=0$), it will have the form 
$X_\ell:U_{\ell-1}\cup\{p_\ell\}\to R$
for some $p_\ell\in u_\ell^{\Aut(U_{\ell-1})}$. Furthermore, $X_\ell$ 
extends the normalized assignment $S_{\ell-1}=X|_{U_{\ell-1}}$. 
Execute the test (T1') on $X_\ell$ and $S_{\ell-1}$. 
If the test (T1') fails, reject the subtree of $X_\ell$ from further 
consideration. If either $\ell=0$ or the test (T1') passes, then normalize 
$X_\ell$ to obtain $S_\ell=X_\ell^{\nu(p_\ell)}$ with $S_\ell:U_\ell\to R$.
At this point $S_\ell$ has been accepted as the unique representative of its
isomorphism class. If $\ell=k$, then output the full prefix assignment $S_\ell$.
If $\ell\leq k-1$, proceed to consider extensions of $S_\ell$ at level
$\ell+1$ as follows. 
Iterate over each variable-value pair $(p_{\ell+1},r)$ with 
$p_{\ell+1}\in u_{\ell+1}^{\Aut(U_{\ell})}$ and $r\in R$. 
Construct the assignment $X_{\ell+1}: U_\ell\cup\{p_{\ell+1}\} \to R$ 
by setting $X_{\ell+1}(p_{\ell+1})=r$ and $X_{\ell+1}(u)=S_\ell(u)$ 
for all $u\in U_\ell$. For each constructed $X_{\ell+1}$, perform 
the test (T2'). If the test (T2') passes, push $X_{\ell+1}$ to the stack. 
\end{tabular}
\end{center}
We observe that the procedure (W) above implements procedure (P') using
the stack to maintain the state of the search. In particular, when a single
worker process executes the search, we obtain a standard depth-first 
traversal of the search tree. However, we also observe that procedure (W) 
pushes {\em all} the child nodes of $S_\ell$ to the stack before consulting
the stack for further work. This enables multiple worker processes, 
all executing procedure (W), to work in parallel, if we take care to ensure 
that (i) push and pop operations to the stack are atomic, and (ii)
the termination condition is changed from the stack being empty to the stack 
being empty and all worker procedures being idle. Furthermore, as 
presented in more detail in what follows, we can distribute the stack across 
multiple compute nodes by appropriately communicating push and pop requests 
between nodes.

\subsection{Parallelization and distributing the stack}
\label{sect:parallelization}
We parallelize the search using the
OpenMPI implementation~\cite{Gabriel:2004} of the Message Passing
Interface (MPI)~\cite{MPI:2015,Pacheco:1997,Gropp:2014}.
We provide two different communication modes, both of which rely on a
\emph{master--slave} paradigm with $N$ processes. The {\em master}
process with rank 0 distributes the work to
$N-1$ {\em worker} processes that, in turn, communicate their results 
back to the master process. We now proceed with a more detailed
description of the two communication modes. 

\medskip
\noindent
{\em Master stack mode.}
In the simpler of the two modes, the master process stores the entire
stack. The worker processes interact with the master directly,
making push and pop requests to the master process via MPI messages. 
While inefficient in terms of communication and in terms of potentially 
overwhelming the
master node, this mode provides load balancing that is empirically
adequate for a small number of compute nodes and instances whose
search tree is not too wide.

\medskip
\noindent
{\em Hierarchical stack mode.}
The hierarchical stack mode divides the $N-1$ worker nodes into $M$ classes, 
each of which is associated with a subset of levels of the search tree. 
Each worker process maintains a {\em local stack} for nodes at their
respective levels. Whenever a worker process pushes an assignment,
the assignment is stored in the local stack if the level of the 
assignment belongs to the levels associated with the node; otherwise, 
the assignment is communicated to the master process which then pushes 
the assignment to the {\em global stack} maintained in the master process.
Whenever a worker process pops an assignment, the worker process first 
consults its local stack and pops the assignment from the local stack if
an assignment is available; otherwise, the worker process makes a pop 
request to the master process, which supplies an assignment from the global
stack as soon as an assignment of one of the levels associated with the
worker becomes available. This strategy helps in cases where the search tree
becomes very wide; in our experiments, we found that a simple thresholding
into one low-level process that processes levels $1,2,\ldots,t$, and
$N-2$ high level processes that process levels $t+1,t+2,\ldots,k$ was
sufficient. 

\medskip
For both modes of communication, the master process keeps track of the 
worker processes that are idle, that is, workers that have sent pop 
requests that have not been serviced. If all workers are idle and the 
global stack is empty, the master process instructs all worker processes 
to exit and then exits itself. 

These communication modes serve as a proof-of-concept of the practical 
parallelizability of our present technique for symmetry reduction. 
For parallelization to very large compute clusters, we expect that more 
advanced communication strategies will be required 
(see, for example,~\cite{Dinan:2009,Pezzi:2007} or~\cite{Pacheco:1997}); 
however, the implementation of such strategies is beyond the scope of 
the present work.

\section{Experiments}

\label{sect:experiments}

This section documents an experimental evaluation of our parallel 
implementation of the adaptive prefix-assignment technique. Our main
objective is to demonstrate the effective parallelizability of the approach, 
but we will also report on experiments comparing the performance of our tool 
(without parallelization) with existing tools that do not parallelize. 

\subsection{Instances}
\label{sect:instances}

Let us start by defining the families of input instances used in our 
experiments. First, we study the usefulness of an auxiliary 
symmetry graph with systems of polynomial equations aimed at discovering the 
tensor rank of a small $m\times m\times m$ tensor $T=(t_{ijk})$ modulo 2, with
$t_{ijk}\in\{0,1\}$ and $i,j,k=1,2,\ldots m$.  Computing the rank of a
given tensor is NP-hard \cite{Hastad:1990}.%
\footnote{Yet considerable interest exists to determine tensor ranks
  of small tensors, in particular tensors that encode and enable fast
  matrix multiplication algorithms;
  cf.~\cite{Alekseev:2014,Alekseev:2015,Alekseev:2013,Alekseyev:1985,Blaser:1999,Blaser:2003,Courtois:2012,Hopcroft:1971,Laderman:1976,Strassen:1969,Winograd:1971}. For numerical work on discovering small low-rank tensor decompositions, cf.~\cite{Benson:2015,Huang:2017,Smirnov:2013}.}{}
In precise terms, we seek to find the minimum $r$ such that there
exist three $m\times r$ matrices $A,B,C\in\{0,1\}^{m\times r}$ such
that for all $i,j,k=1,2,\ldots,m$ we have
\begin{equation}
\label{eq:tensor-rank}
\sum_{\ell=1}^r a_{i\ell}b_{j\ell}c_{k\ell}=t_{ijk}\pmod 2\,. 
\end{equation}
Such instances are easily compilable into CNF with $A,B,C$
constituting three matrices of Boolean variables so that the task
becomes to find the minimum $r$ such that the compiled CNF instance is
satisfiable. Independently of the target tensor $T$, such instances
have a symmetry group of order at least $r!$ due to the fact that the
columns of the matrices $A,B,C$ can be arbitrarily permuted so that
\eqref{eq:tensor-rank} maps to itself.  In our experiments, we select
the entries of $T$ uniformly at random so that the number of $1$s in
$T$ is exactly $n$. We use the first three rows of the matrix A as the
prefix sequence.

As a further family of instances with considerable symmetry, we study
the \emph{Clique Coloring Problem} (CCP) that yields empirically
difficult-to-solve instances for contemporary SAT
solvers~\cite{Manthey:2014}. For positive integer parameters $n$, $s$,
and $t$, the CCP asks whether there exists an undirected $t$-colorable
graph on $n$ nodes such that the graph contains a complete graph $K_s$
as a subgraph. Such instances are unsatisfiable if $s>t$. The
particular encoding that we use (see~\cite{Manthey:2014}) is as
follows. Introduce variables $x_{i,j}$ for $1\leq i,j \leq n$ with
$i\neq j$ to indicate the presence of an edge joining vertex $i$ and
$j$, variables $y_{p,j}$ for $1\leq p \leq s$ with $1\leq j\leq n$ to
indicate that vertex $j$ occupies slot $p$ in a clique, and variables
$z_{i,k}$ for $1\leq i \leq n$ and $1 \leq k \leq t$ to indicate that
vertex $i$ has color $k$. The clauses are
\begin{enumerate}
\item $\bigwedge_{1\leq p \leq s} \bigvee_{1 \leq j \leq n} y_{p,j}$\,,
\item $\bigwedge_{1\leq p\leq s} \bigwedge_{1\leq q \leq s : p\neq q} \bigwedge_{1 \leq j \leq n} \overline{y_{p,j}} \vee \overline{y_{q,j}}$\,,
\item $\bigwedge_{1\leq p\leq s} \bigwedge_{1\leq q \leq s : p\neq q} \bigwedge_{1 \leq i \leq n} \bigwedge_{1 \leq j \leq n : i \neq j} \overline{y_{p,i}} \vee \overline{y_{q,j}} \vee x_{i,j}$\,,
\item $\bigwedge_{1\leq k\leq t} \bigwedge_{1 \leq i \leq n} \bigwedge_{1 \leq j \leq n : i \neq j} \overline{z_{i,k}} \vee \overline{z_{j,k}} \vee \overline{x_{i,j}}$ \,, and
\item $\bigwedge_{1\leq i \leq n} \bigvee_{1 \leq k \leq t} z_{i,k}$\,.
\end{enumerate}
We consider unsatisfiable instances with parameters $s \in\{5,6\}$,
$t = s-1$, and let $n$ vary from 15 to 20 in the case of $s=5$ and 12
to 24 when $s=6$. We use the variables $y_{1,1},y_{1,2},\ldots,y_{1,n}$ 
as the prefix sequence. The auxiliary graph for encoding the symmetries
is constructed as follows. Introduce a vertex for each variable
$x_{i,j}$, for each variable $y_{p,j}$, and for each variable $z_{i,k}$. 
These vertices are colored with three distinct colors, 
one color for each type of variable. Next, introduce three types of
auxiliary vertices, with each type colored with its own distinct color.
Introduce vertices $1,2,\ldots,n$ for the $n$ nodes, vertices 
$1',2',\ldots,s'$ for the $s$ clique slots, and vertices $1'',2'',\ldots,t''$
for the $t$ node colors. Thus, in total the graph consists
of $n(n-1)+sn+tn+n+s+t$ vertices colored with six distinct colors. 
To complete the construction of the auxiliary graph, 
introduce edges to the graph so that each variable $x_{i,j}$ is joined to 
the nodes $i$ and $j$, each variable $y_{p,j}$ is joined to clique slot $p'$ 
and to the node $j$, and each variable $z_{i,k}$ is joined to the node $i$ 
and to the node color $k''$.

We study the parallelizability of our algorithm using two input instances
with hard combinatorial symmetry. 
The first instance, which we call $R(4,4;18)$ in what follows, is 
an unsatisfiable CNF instance that asks whether there exists an 18-node graph 
with the property that neither the graph nor its complement contains the 
complete graph $K_4$ as a subgraph. That is, we ask whether the Ramsey number 
$R(4,4)$ satisfies $R(4,4) > 18$ (in fact, $R(4,4)=18$~\cite{Graham:1990}).  
No auxiliary graph is provided to accompany this instance. 
The second instance consists of an empty CNF over 36 variables together with
an auxiliary graph that encodes the isomorphism classes of 9-node graphs by
inserting a variable vertex in the middle of each of the $\binom{9}{2}=36$ 
edges of the complete graph $K_9$. Applying \texttt{reduce} with 
a length-36 prefix sequence (listing the 36 variable vertices in any order) 
yields a complete listing of all the 274668 isomorphism classes of 9-node 
graphs. The number of isomorphism classes of graphs of order $n$ is the 
sequence A000088 in the {\em Online Encyclopedia of Integer Sequences}.

\subsection{Hardware and software configuration}
The experiments were performed on a cluster of Dell PowerEdge C4130
compute nodes, each equipped with two Intel Xeon E5-2680v3 CPUs 
(12 cores per CPU, 24 cores per node) and 128 GiB (8$\times$16 GiB) of DDR4-2133 main memory, 
running the CentOS 7 distribution of GNU/Linux.
Comparative experiments were executed by allocating a single core on a
single CPU of a compute node.  All experiments were conducted as batch jobs
using the \texttt{slurm} batch scheduler, and running between one to four 
physical nodes, with one to 24 cores allocated in each node, using one 
MPI process per core.  OpenMPI version 2.1.1 was used as the MPI implementation.

\subsection{Symmetry reduction tools and SAT solvers}
We report on three methods for symmetry reduction:
(1) no reduction (``\texttt{raw}''), 
(2) \texttt{breakid} version 2.1-152-gb937230-dirty\footnote{We thank
  Bart Bogaerts for implementing custom graph input in \texttt{breakid}.}~\cite{DevriendtEtAl:SAT2016},
(3) our technique (``\texttt{reduce}'') with a user-selected prefix.
Three different SAT solvers were used in the experiments:
\texttt{lingeling} and \texttt{ilingeling} version
\texttt{bbc-9230380}~\cite{Biere:SATCOMP2016}, 
and \texttt{glucose} version
4.1~\cite{Audemard:2016}. 
We use the incremental solver \texttt{ilingeling} together with the
incremental CNF output of \texttt{reduce}. 

\subsection{Experiments on parallel speedup}
\label{sect:parallelexperiments}

This section documents experiments that study the wall-clock running time 
of symmetry reduction using our tool \texttt{reduce} as we increase the number
of CPU cores and compute nodes participating in parallel 
symmetry reduction. The range of the experiments was between one to four 
compute nodes, with one to 24 cores allocated in each node. 
One MPI process was launched per core. Each node was exclusively reserved 
for the experiment. In addition to the wall-clock running time, we 
measure the {\em total reserved time} that is obtained by recording, 
for each core, the length of the time interval the core is reserved 
for an experiment, and taking the sum of these time intervals. 
The total reserved time conservatively tracks the total resources consumed 
by an experiment in a batch job environment regardless of whether each 
allocated core is running or idle. 

The results of our parallel speedup experiments are displayed in
Figure~\ref{fig:mpi}. The top-left plot in the figure displays the
parallel speedup (ratio of parallel wall-clock running time to
sequential running time) of running our tool \texttt{reduce} on the
instance $R(4,4;18)$ with a length-33 prefix sequence as a function
of the number of cores used for one, two, and four allocated compute
nodes. We also display the line $y=x$ for reference to compare against
perfect linear speedup.  As the number of cores grows, in the top-left
plot we observe linear scaling of the speedup as a function of the
number of cores. The slope of the speedup yet remains somewhat short
of the perfect $y=x$ scaling. This is most likely due to the use of
the master stack mode and associated communication overhead.  
The top-right plot displays the total reserved time to demonstrate the total
resource usage in addition to the parallel speedup. 
Table~\ref{tab:r4418nodecount} displays the number of canonical
partial assignments at different levels of the search tree explored 
by \texttt{reduce}. 

The two plots in the middle row of Figure~\ref{fig:mpi} display the
parallel speedup and the total reserved time of executing our tool
\texttt{reduce} on the instance A000088 (with $n=9$ and a length-36
prefix sequence) in the master stack mode. This instance requires
extensive stack access with many easy instances of canonical labeling
(cf.~Table~\ref{tab:a000088nodecount} and compare with 
Table~\ref{tab:r4418nodecount}); accordingly we observe poor speedup from 
parallelization in the master stack mode. The two plots in the bottom row of
Figure~\ref{fig:mpi} show an otherwise identical experiment but now
executed in hierarchical stack mode with the threshold parameter set to
$t=21$, in which case both the parallel speedup obtained and the total 
resource usage become substantially better.

When the number of processes is small, Figure~\ref{fig:mpi} reveals 
inefficiency in terms of the total reserved time compared with a larger
number of processes. This inefficiency is explained by two factors. 
First, when the number of processes is small,
a significant fraction of the total reserved time is used by the master process
which does not contribute work to the exploration of the search tree
but does consume reserved time from the start to the end of the computation. 
As soon as more worker processes start exploring the search tree, the 
total reserved time decreases because the time consumed by the master process 
decreases. Second, in hierarchical stack mode, a small number of processes 
means that some of the worker nodes processing lower levels of the search 
tree can run out of work---but will still consume total reserved time---as 
assignments in these levels are exhausted, while the small number of processes 
assigned to work on the higher levels of the tree still remain at work. 
This bottleneck can be alleviated by increasing the number of workers 
associated with the higher levels.

\begin{figure}
  \begin{tabular}{cc}
    \includegraphics[width=0.45\textwidth]{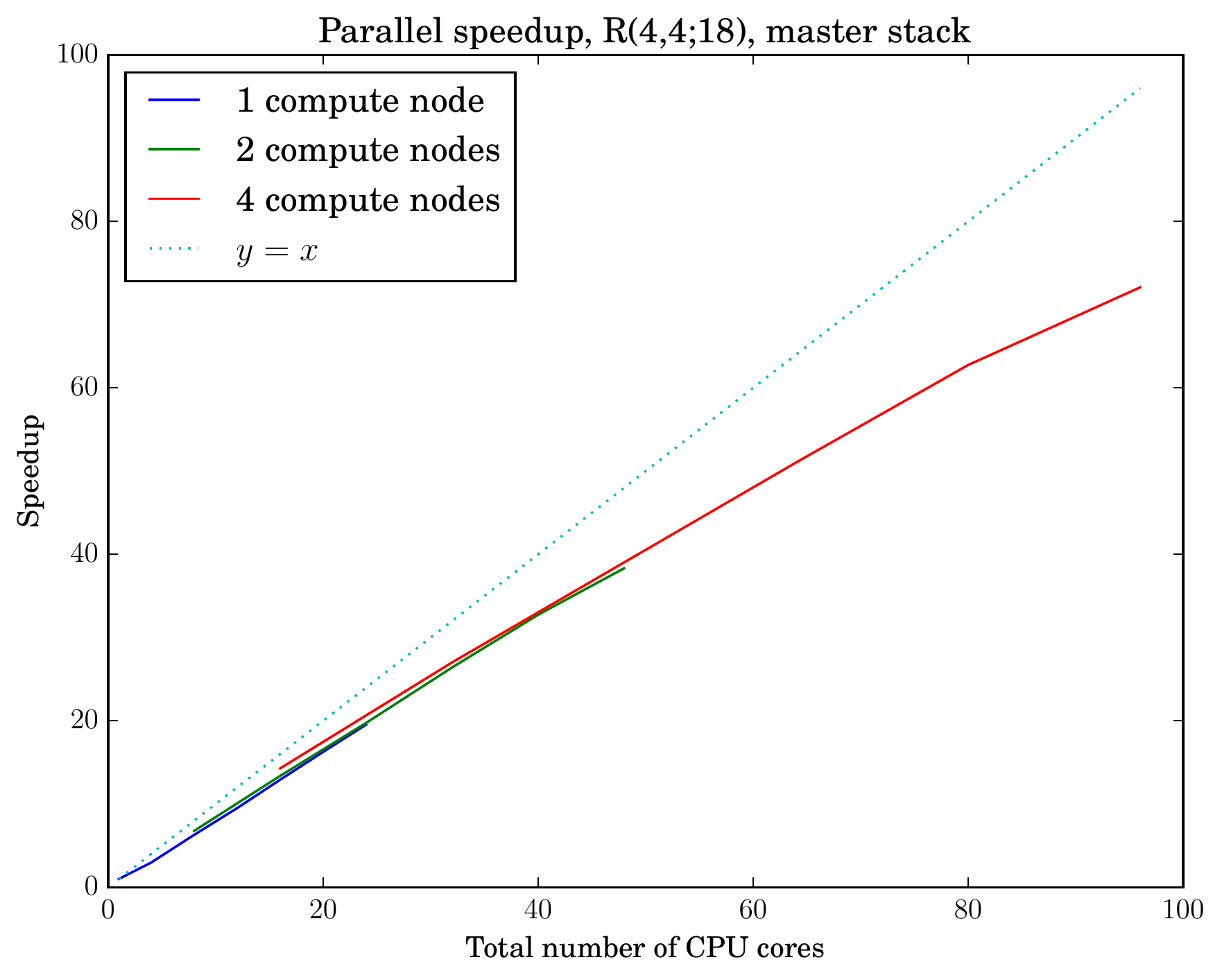} &
    \includegraphics[width=0.45\textwidth]{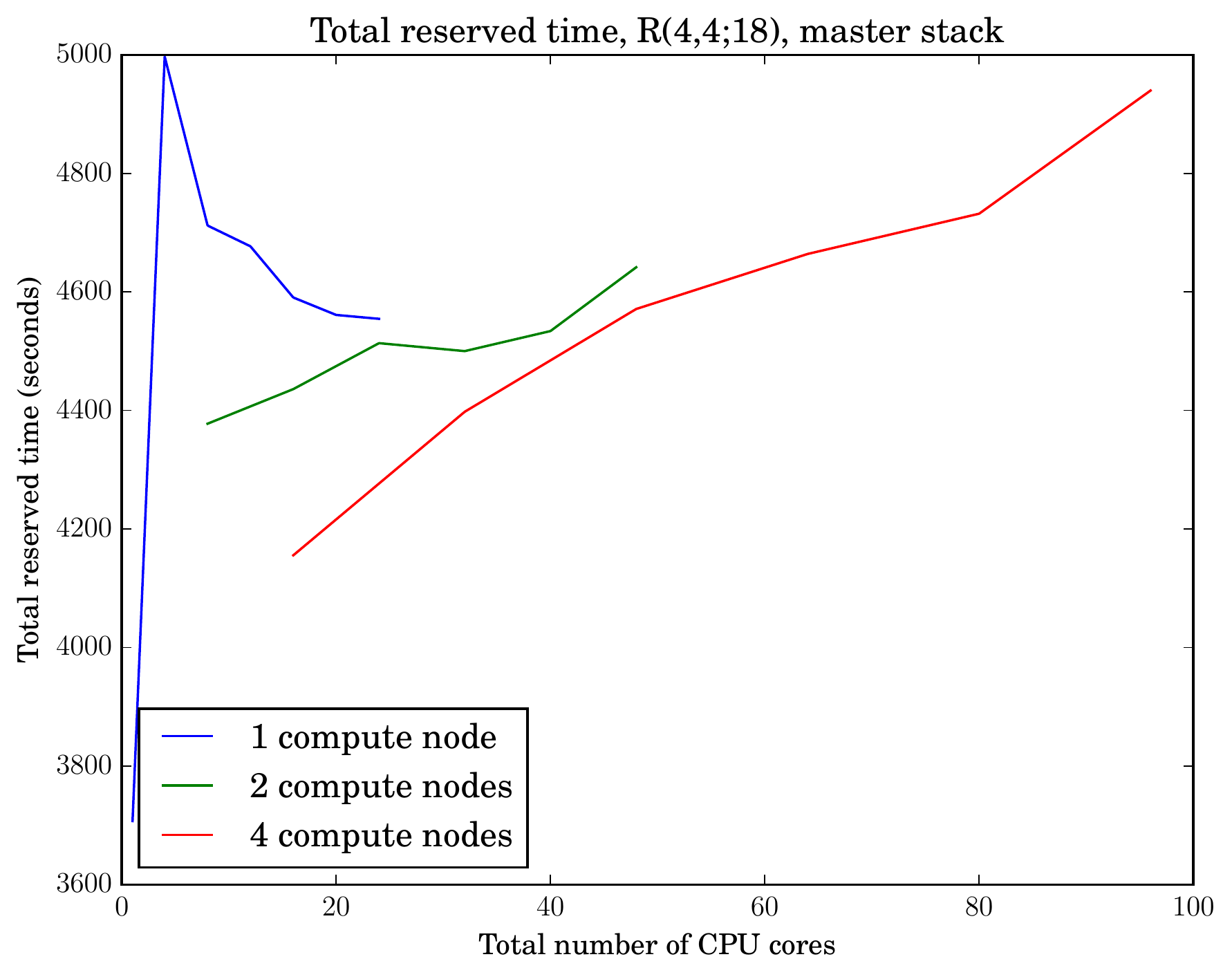} \\
    \includegraphics[width=0.45\textwidth]{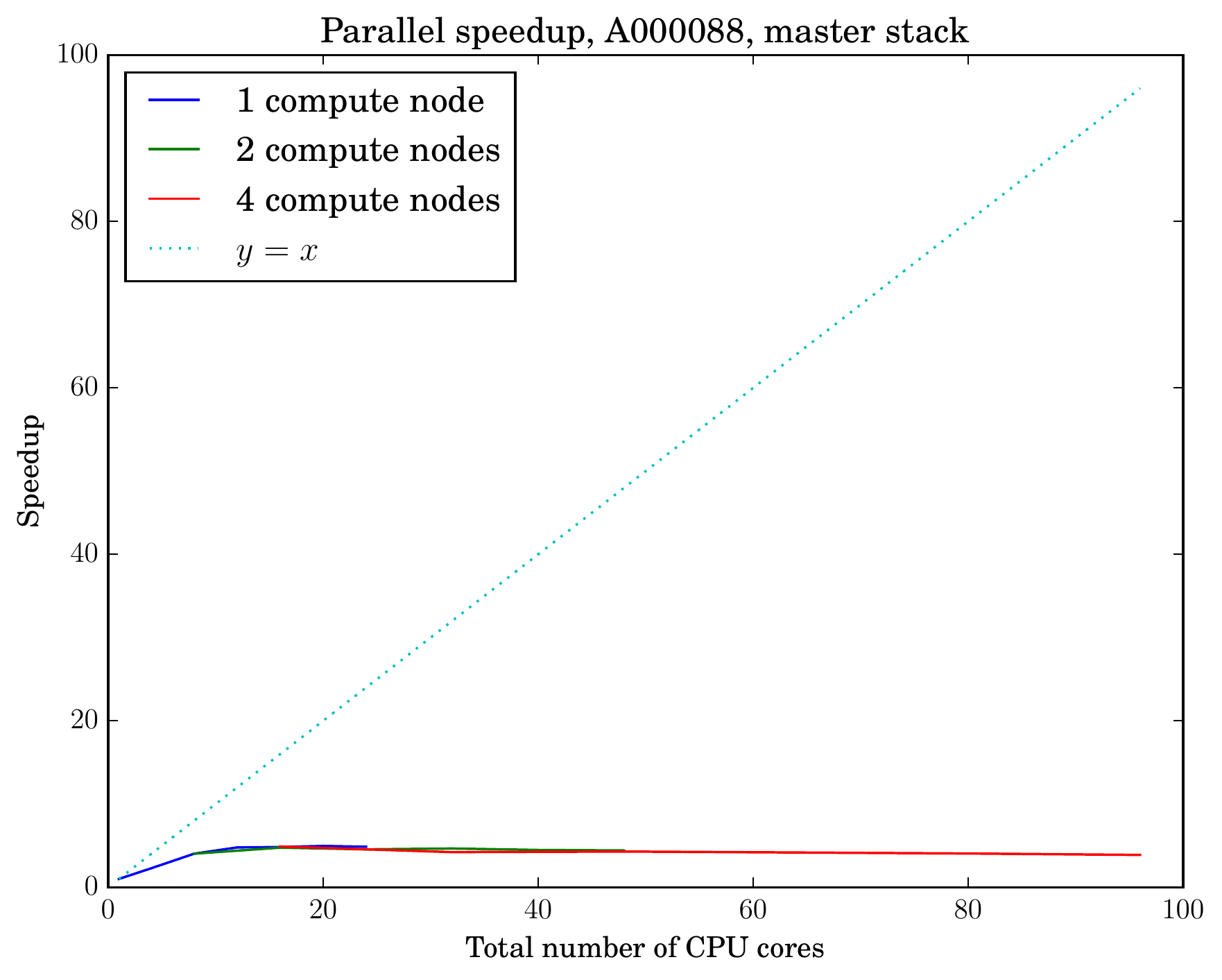} &
    \includegraphics[width=0.45\textwidth]{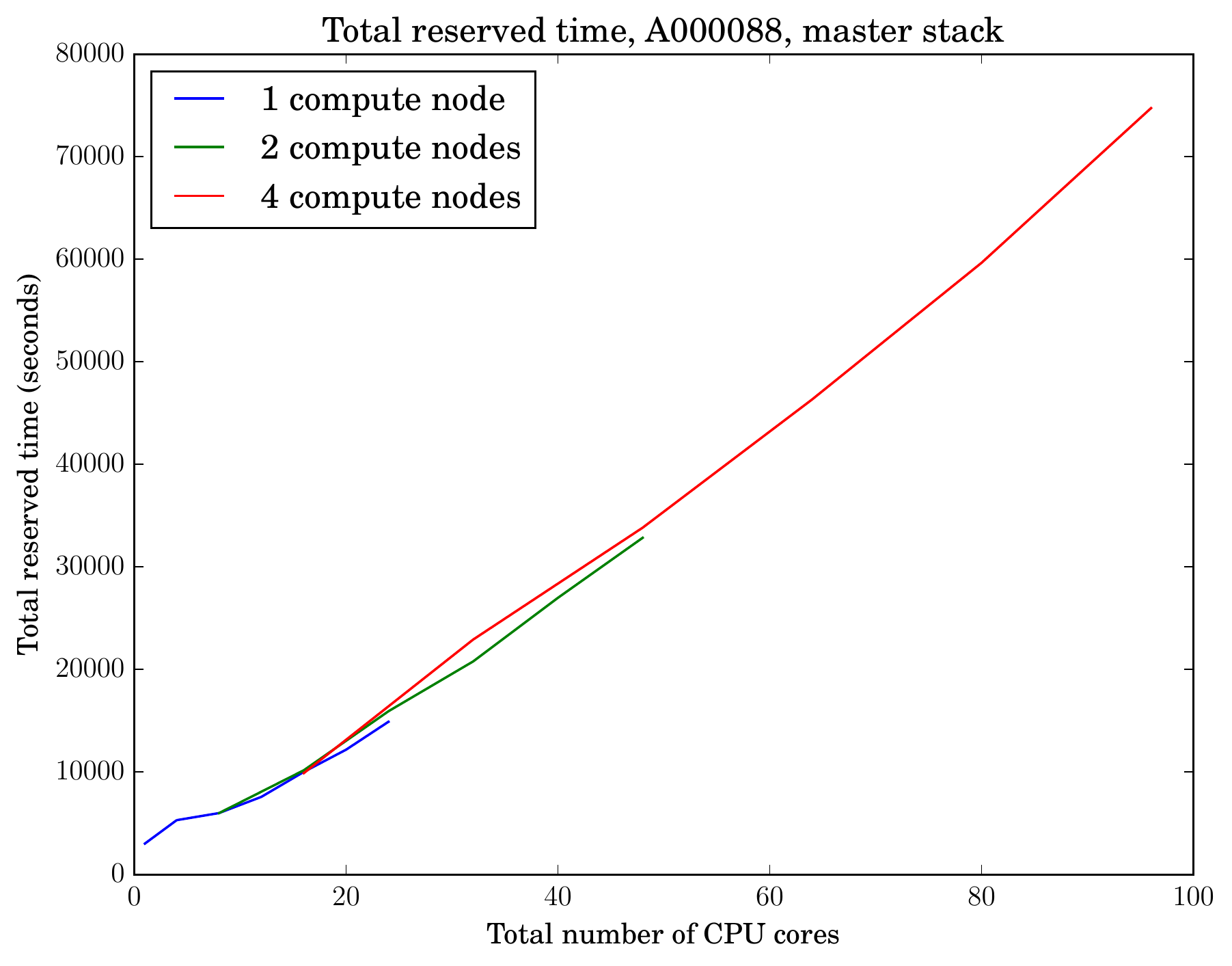} \\
    \includegraphics[width=0.45\textwidth]{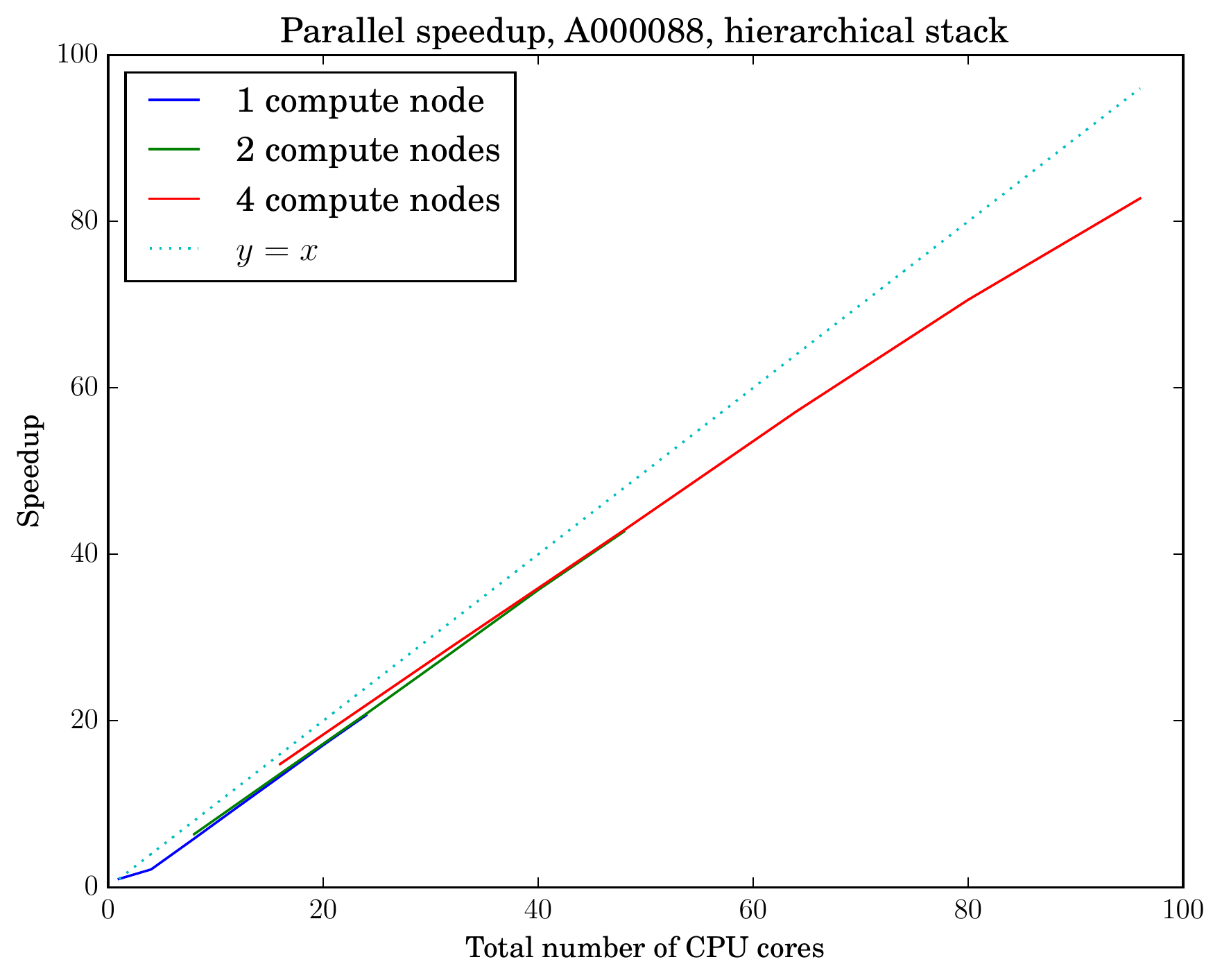} &
    \includegraphics[width=0.45\textwidth]{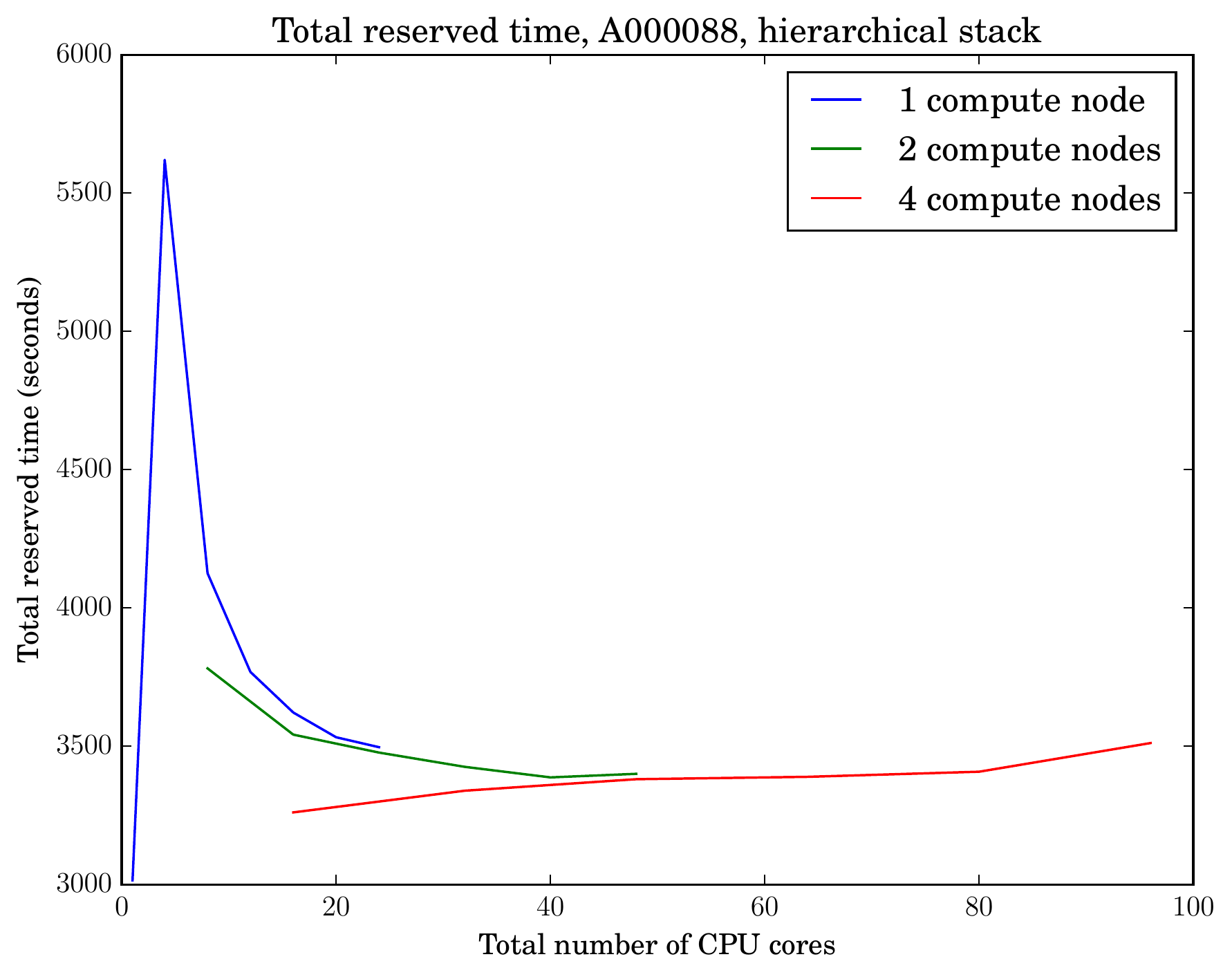} 
  \end{tabular}
  \caption{The plots on the left display the parallel speedup factor
    obtained with increasing number of cores. The plots on the right
    display the total reserved time in seconds of all the cores with
    increasing number of cores. The first row shows the instance
    $R(4,4;18)$, with parallelization executed using the master stack
    mode of communication.  We observe solid parallel speedup with
    increasing number of cores. The second row shows the instance
    A000088 using the master stack mode of communication. Here the
    speedup is unsatisfactory due to extensive accesses to the master
    stack caused by a wide search tree with many easy instances of 
    canonical labeling (cf.~Table~\ref{tab:a000088nodecount}). 
    This bottleneck can be alleviated by to the hierarchical
    stack mode. The third row shows the instance A000088 executed the
    hierarchical stack mode. Now we observe solid parallel speedup
    with increasing number of cores. The peaks in the total reserved
    time for a small number of cores are caused by an unbalanced work 
    allocation between the cores that eases with increasing number of
    cores; see Section~\ref{sect:parallelexperiments}.}
  \label{fig:mpi}
\end{figure}

\begin{table}
  \caption{The number of canonical partial assignments in the instance $R(4,4;18)$ at different levels of the search tree explored by \texttt{reduce}.}
  \label{tab:r4418nodecount}
  \begin{tabular}{|rr|c|rr|c|rr|}
    \cline{1-2}\cline{4-5}\cline{7-8}
  level & assignments   &&         level & assignments   &&  level & assignments\\\hhline{==~==~==}
1	&	2	&&	12	&	13	&&	23	&	1848	\\
2	&	3	&&	13	&	14	&&	24	&	2400	\\
3	&	4	&&	14	&	15	&&	25	&	2970	\\
4	&	5	&&	15	&	16	&&	26	&	3520	\\
5	&	6	&&	16	&	17	&&	27	&	4004	\\
6	&	7	&&	17	&	18	&&	28	&	4368	\\
7	&	8	&&	18	&	96	&&	29	&	4550	\\
8	&	9	&&	19	&	300	&&	30	&	4480	\\
9	&	10	&&	20	&	560	&&	31	&	4080	\\
10	&	11	&&	21	&	910	&&	32	&	3264	\\
11	&	12	&&	22	&	1344	&&	33	&	1050	\\\cline{1-2}\cline{4-5}\cline{7-8}
  \end{tabular}
\end{table}

\begin{table}[h!]
  \caption{The number of canonical partial assignments in the instance A000088 at different levels of the search tree explored by \texttt{reduce}. We observe that the search tree is considerably wider at the intermediate levels compared with the last level.}{}
  \label{tab:a000088nodecount}
  \begin{tabular}{|rr|c|rr|c|rr|}
    \cline{1-2}\cline{4-5}\cline{7-8}
  level & assignments   &&         level & assignments   &&  level & assignments\\\hhline{==~==~==}
    1	&	2	&&	13	&	336	&&	25	&	346376	\\
2	&	3	&&	14	&	336	&&	26	&	47418	\\
3	&	4	&&	15	&	140	&&	27	&	644016	\\
4	&	5	&&	16	&	1216	&&	28	&	3256288	\\
5	&	6	&&	17	&	5256	&&	29	&	4336496	\\
6	&	7	&&	18	&	9936	&&	30	&	508140	\\
7	&	8	&&	19	&	13664	&&	31	&	5245032	\\
8	&	9	&&	20	&	13104	&&	32	&	19768096	\\
9	&	42	&&	21	&	2676	&&	33	&	2409488	\\
10	&	120	&&	22	&	34500	&&	34	&	13814848	\\
11	&	200	&&	23	&	183120	&&	35	&	4147832	\\
12	&	280	&&	24	&	328032	&&	36	&	274668	\\\cline{1-2}\cline{4-5}\cline{7-8}
  \end{tabular}
\end{table}

\subsection{Experiments comparing with other tools}

We compared our present tool \texttt{reduce} against the 
tool \texttt{breakid}~\cite{DevriendtEtAl:SAT2016}. 
Since \texttt{breakid} does not parallelize, no parallelization was used 
in these experiments and all experiments were executed using a single 
compute core. All running times displayed in the tables that follow 
are in seconds, with ``t/o'' indicating a time-out of 25 hours of wall-clock 
time. Other compute load was in general present on the compute nodes where 
these experiments were run. 

Table \ref{tbl:equationswgraph} shows the results of a tensor rank 
computation modulo 2 for two random tensors $T$ with $m=5$, $n=9$ and $m=5$, 
$n=20$ with (top table) and without (bottom table) an auxiliary graph. 
When $m=5$ and $n=9$, the tensor has rank 8 and decompositions for rank 7 
and 8 are sought. When $m=5$ and $n=20$, the tensor has rank 9 and 
decompositions of rank 8 and 9 are sought. For both tensors we observe 
decreased running time due to symmetry reduction. Comparing the top and
bottom tables, we observe the relevance of the
graph representation of the symmetries in \eqref{eq:tensor-rank},
which are not easily discoverable from the compiled CNF. 
As the auxiliary graph, we used the graph representation of the 
system \eqref{eq:tensor-rank}, constructed as in Example~\ref{exa:brent}.

\begin{table}
\caption{Comparing different tools for preprocessing and then solving instances with hard symmetry not easily discoverable from a compiled CNF encoding. Here the instances ask for modulo-2 tensor decompositions for two $5\times 5\times 5$ tensors with (top) and without (bottom) an auxiliary graph. We observe that the auxiliary graph gives a marked improvement in the running times of preprocessing, making all the instances tractable. Without the auxiliary graph to highlight the symmetry, the two unsatisfiable instances are intractable within the timeout threshold of 90,000~seconds. All running times are in seconds.}
\label{tbl:equationswgraph}
\footnotesize
\begin{tabular}{|@{\,}r@{\ }r@{\ }r@{\,}|@{\,}r@{\ }r@{\,}|r@{\,}|@{\,}r@{\ }r@{\,}|@{\ }r@{\ }|@{\ }r@{\,}r@{\,}r@{\ }|@{\,}l@{\ }|}
\hline
\multicolumn{13}{|c|}{with auxiliary graph}\\[1mm]
  & & & \multicolumn{2}{c|}{raw} & prep. & \multicolumn{2}{c@{\,}|@{\ }}{breakid} & prep. & \multicolumn{3}{c@{\,}|@{\,}}{reduce} & \\
$m$ & $r$ & $n$ & glucose & lingeling & breakid & glucose & lingeling & reduce & glucose & lingeling & ilingeling & Sat?
\\\hline\hline
5 & 7 & 9 & t/o & t/o & 0.28 & 30.07 & 30.73 & 87.35 & 77.87 & 66.40 & 47.67 & No \\
5 & 8 & 9 & 0.36 & 3.91 & 0.63 & 0.33 & 5.61 & 290.70 & 1.69 & 13.93 & 0.50 & Yes \\
5 & 8 & 20 & t/o & t/o & 0.61 & 1078.54 & 1273.49 & 290.78 & 2641.97 & 7699.27 & 885.01 & No \\
5 & 9 & 20 & 1.44 & 1.28 & 1.68 & 72.09 & 26.33 & 881.85 & 228.06 & 371.09 & 40.57 & Yes \\
\hline
\end{tabular}\\[2mm]
\footnotesize
\begin{tabular}{|@{\,}r@{\ }r@{\ }r@{\,}|@{\,}r@{\ }r@{\,}|r@{\,}|@{\,}r@{\ }r@{\,}|@{\ }r@{\ }|@{\ }r@{\,}r@{\,}r@{\ }|@{\,}l@{\ }|}
\hline
\multicolumn{13}{|c|}{without auxiliary graph}\\[1mm]
  & & & \multicolumn{2}{c|}{raw} & prep. & \multicolumn{2}{c@{\,}|@{\ }}{breakid} & prep. & \multicolumn{3}{c@{\,}|@{\,}}{reduce} & \\
$m$ & $r$ & $n$ & glucose & lingeling & breakid & glucose & lingeling & reduce & glucose & lingeling & ilingeling & Sat?
\\\hline\hline
5 & 7 & 9 & t/o & t/o & 0.64 & t/o & t/o & t/o & n/a & n/a & n/a & No \\
5 & 8 & 9 & 0.36 & 3.91 & 0.60 & 0.40 & 1.22 & t/o & n/a & n/a & n/a & Yes \\
5 & 8 & 20 & t/o & t/o & 0.32 & t/o & t/o & t/o & n/a & n/a & n/a & No \\
5 & 9 & 20 & 1.44 & 1.28 & 1.59 & 3.30 & 15.10 & t/o & n/a & n/a & n/a & Yes \\
\hline
\end{tabular}
\end{table}

Table~\ref{tbl:ccp} shows the results of applying \texttt{breakid} and our
tool \texttt{reduce} as preprocessors for solving instances of 
the Clique Coloring Problem. We observe that for sufficiently large instances,
our tool is faster than \texttt{breakid} in the combined runtime of 
preprocessor and solver.

Table~\ref{tbl:ccpgraphcomparison} compares running times of \texttt{reduce}
on instances of the Clique Coloring Problem (i) using 
the graph automatically constructed from CNF, and (ii) using a tailored
auxiliary graph constructed as described in Section~\ref{sect:instances}. 
For these instances, the available symmetry can be easily discovered 
directly from the CNF encoding, but we observe that the use of the tailored 
auxiliary graph does result in faster preprocessing times for \texttt{reduce}.

\begin{table}
\caption{Comparing different tools for preprocessing and then solving instances with hard symmetry. Here the instances ask for solutions to the Clique Coloring Problem with parameters indicated on the left. On this family of instances, our tool \texttt{reduce} is faster than the tool \texttt{breakid} for sufficiently large parameters. All instances are unsatisfiable. All running times are in seconds. A timeout threshold of 90,000 seconds was applied.}
\label{tbl:ccp}
\footnotesize
\begin{tabular}{|@{\,}r@{\ }r@{\ }r@{\,}|@{\,}r@{\ }r@{\,}|r@{\,}|@{\,}r@{\ }r@{\,}|@{\ }r@{\ }|@{\ }r@{\ }r@{\ }r|}
\hline
  & & & \multicolumn{2}{c|}{raw} & prep. & \multicolumn{2}{c@{\,}|@{\ }}{breakid} & prep. & \multicolumn{3}{c|}{reduce} \\
$n$ & $s$ & $t$ & glucose & lingeling & breakid & glucose & lingeling & reduce & glucose & lingeling & ilingeling
\\\hline\hline
15 & 5 & 4 & 722.32 & 811.59 & 2.82 & 126.35 & 154.27 & 19.34 & 19.82 & 31.61 & 32.92 \\
16 & 5 & 4 & 1038.88 & 1839.05 & 10.68 & 238.52 & 570.19 & 8.07 & 25.62 & 39.55 & 38.31 \\
17 & 5 & 4 & 4481.54 & 8865.19 & 1.94 & 597.35 & 498.20 & 12.81 & 105.35 & 57.96 & 54.66 \\
18 & 5 & 4 & 2709.96 & 4762.23 & 9.96 & 559.86 & 460.70 & 14.11 & 40.68 & 74.71 & 66.19 \\
19 & 5 & 4 & 6701.65 & 6819.77 & 10.66 & 586.81 & 651.10 & 19.13 & 107.73 & 106.62 & 85.38 \\
20 & 5 & 4 & 8901.20 & 7777.35 & 1.15 & 1294.31 & 1579.77 & 25.37 & 157.91 & 134.55 & 248.56 \\
12 & 6 & 5 & 38835.75 & 15517.28 & 9.87 & 1602.96 & 745.83 & 2.42 & 1190.58 & 677.38 & 751.27 \\
13 & 6 & 5 & 26017.87 & 50312.82 & 9.91 & 7032.11 & 3506.61 & 9.68 & 1439.84 & 1440.30 & 1420.10 \\
14 & 6 & 5 & t/o & t/o & 2.28 & 8417.69 & 5384.79 & 17.18 & 2360.88 & 5559.26 & 2543.99 \\
15 & 6 & 5 & t/o & t/o & 9.05 & 10537.53 & 7316.61 & 6.49 & 3504.82 & 4104.10 & 4140.57 \\
16 & 6 & 5 & t/o & t/o & 9.46 & 41355.16 & 27699.48 & 30.52 & 6858.69 & 5612.36 & 5708.18 \\
17 & 6 & 5 & t/o & t/o & 0.94 & t/o & t/o & 11.48 & 11329.16 & 20597.16 & 10460.73 \\
18 & 6 & 5 & t/o & t/o & 5.34 & t/o & t/o & 23.81 & 17347.43 & 52703.19 & 16873.36 \\
19 & 6 & 5 & t/o & t/o & 7.18 & t/o & t/o & 82.93 & 29689.84 & 19969.09 & 21195.04 \\
20 & 6 & 5 & t/o & t/o & 3.38 & t/o & t/o & 29.36 & 76600.29 & 35850.94 & 27035.80 \\
21 & 6 & 5 & t/o & t/o & 1.50 & t/o & t/o & 148.64 & t/o & 45963.02 & 48542.48 \\
22 & 6 & 5 & t/o & t/o & 1.74 & t/o & t/o & 198.15 & t/o & 61414.88 & 66279.68 \\
23 & 6 & 5 & t/o & t/o & 1.91 & t/o & t/o & 267.67 & t/o & t/o & 78463.34 \\
\hline
\end{tabular}
\end{table}

\begin{table}
\caption{The effect of a tailored auxiliary graph on the running time of \texttt{reduce} when applied on instances of the Clique Coloring Problem with parameters indicated on the left. We see that there is a marked decrease in the running times when the auxiliary graph is available. All running times are in seconds.}
\label{tbl:ccpgraphcomparison}
\footnotesize
\begin{tabular}{|@{\,}r@{\ }r@{\ }r@{\,}|@{\,}r@{\ \ \ }r@{\,}|}
\hline
    &     &     & \multicolumn{2}{c|}{reduce}\\[1mm]
    &     &     & with  & without\\
$n$ & $s$ & $t$ & graph & graph
\\\hline\hline
15 & 5 & 4 & 19.34 & 399.25 \\
16 & 5 & 4 & 8.07 & 158.47 \\
17 & 5 & 4 & 12.81 & 223.24 \\
18 & 5 & 4 & 14.11 & 1151.05 \\
19 & 5 & 4 & 19.13 & 1629.34 \\
20 & 5 & 4 & 25.37 & 586.78 \\
12 & 6 & 5 & 2.42 & 139.14 \\
13 & 6 & 5 & 9.68 & 61.46 \\
14 & 6 & 5 & 17.18 & 381.85 \\
15 & 6 & 5 & 6.49 & 587.87 \\
16 & 6 & 5 & 30.52 & 881.09 \\
17 & 6 & 5 & 11.48 & 322.53 \\
18 & 6 & 5 & 23.81 & 450.28 \\
19 & 6 & 5 & 82.93 & 2388.18 \\
20 & 6 & 5 & 29.36 & 851.68 \\
21 & 6 & 5 & 148.64 & 1195.23 \\
22 & 6 & 5 & 198.15 & 1543.51 \\
23 & 6 & 5 & 267.67 & 5545.20 \\
\hline
\end{tabular}
\end{table}

\subsection*{Acknowledgements}
The research leading to these results has received 
funding from the European Research Council under the European Union's 
Seventh Framework Programme (FP/2007-2013) / ERC Grant Agreement 338077 
``Theory and Practice of Advanced Search and Enumeration'' (M.K., P.K., and J.K.). We gratefully acknowledge the use of computational resources provided 
by the Aalto Science-IT project at Aalto University. 
We thank Tomi Janhunen and Bart Bogaerts for useful discussions. 

A preliminary conference abstract of this paper appeared in Junttila
T., Karppa M., Kaski P., Kohonen J. (2017) An Adaptive
Prefix-Assignment Technique for Symmetry Reduction. In: Gaspers S.,
Walsh T. (eds) Theory and Applications of Satisfiability Testing – SAT
2017. SAT 2017. Lecture Notes in Computer Science, vol
10491. Springer, Cham.


\bibliographystyle{amsplain}
\bibliography{paper}


\end{document}